\documentclass[twocolumn,journal]{IEEEtran}
\usepackage{algorithm}
\usepackage{algorithmicx}
\usepackage{algpseudocode}
\usepackage{amsfonts,amsmath,amssymb,amsthm}
\usepackage{multirow,graphicx,cite,color,footnote}
\usepackage{ulem}
\usepackage{tikz}
\newtheorem{Theorem}{Theorem}
\newtheorem{Remark}{Remark}

\newtheorem{Example}{Example}

\newtheorem{Fact}{Fact}

\normalsize

\ifCLASSINFOpdf
\else
\fi

\begin{document}
%
\title{
A Generic Transformation for Optimal Node Repair in  MDS Array Codes over $\mathbf{F}_2$}

%
%
%
\author{Jie Li,~\IEEEmembership{Member,~IEEE,}  Xiaohu Tang,~\IEEEmembership{Senior~Member,~IEEE,} and Camilla Hollanti,~\IEEEmembership{Member,~IEEE,}

\thanks{The work of J. Li was supported in part by the National Science Foundation of China under Grant No. 61801176.  The work of Xiaohu Tang was
supported in part by the National Natural Science Foundation of China
under Grant No. 61941106 and 61871331. The work of C. Hollanti was supported by the Academy of Finland, under Grants No. 318937 and 336005.}
\thanks{J. Li was with the Hubei Key Laboratory of Applied Mathematics, Faculty of Mathematics and Statistics, Hubei University, Wuhan 430062, China, and also with the Department of Mathematics and Systems Analysis, P.O. Box 11100, 
                    Aalto University, FI-00076 Aalto (Espoo),  Finland (e-mail: jieli873@gmail.com).}
                    \thanks{X. Tang is with the Information Security and National Computing Grid Laboratory, Southwest Jiaotong University, Chengdu 610031, China (e-mail: xhutang@swjtu.edu.cn).}
        \thanks{C. Hollanti is with the Department of Mathematics and Systems Analysis, P.O. Box 11100, 
                    Aalto University, FI-00076 Aalto (Espoo),  Finland (e-mail: camilla.hollanti@aalto.fi).}
}

%
%

\markboth{IEEE Transactions on Communications}%
{Submitted paper}
%



\maketitle

\begin{abstract}
For high-rate linear systematic maximum distance separable (MDS)  codes,
most early  constructions could initially optimally repair all the systematic nodes but not all the parity nodes. Fortunately, this issue was  first solved by Li \textit{et al.} in (IEEE Trans. Inform. Theory, 64(9), 6257-6267, 2018), where  a   transformation that can convert any  nonbinary   MDS array code into another one with desired properties was proposed. However, the transformation does not work for binary MDS array codes. In this paper, we address this issue by  proposing  another generic transformation that can convert any $[n, k]$ binary    MDS array code into a new one, which endows any $r=n-k\ge2$ chosen nodes with optimal repair bandwidth and optimal rebuilding access properties, and at the same time, preserves the normalized  repair bandwidth/rebuilding
access for the remaining $k$ nodes under some conditions. As two
immediate  applications, we show that
1) by applying the transformation multiple times, any   binary     MDS array code can be converted into one with    optimal rebuilding access for all nodes, 2) any binary MDS array code with optimal repair bandwidth
or  optimal rebuilding access for the systematic nodes    can
be converted into one with the corresponding 
optimality property for all nodes.
\end{abstract}


%
\IEEEpeerreviewmaketitle

\section{Introduction}\label{sec:Intro}
%
%
%
%
\IEEEPARstart{I}{n}   distributed storage systems, one of the major concerns  is  reliability. A common way to fulfill reliability is by introducing redundancy.
Normally, MDS codes such as Reed-Solomon (RS) codes \cite{RS-codes},  can offer maximum reliability for a given storage overhead,   thus they have been
used extensively as the basis for RAID systems and  distributed storage systems  \cite{RS-codes,EVENODD,Blaum_Bruck_Vardy96,EVENODDG,STAR}.

Upon failure of a single storage node, a self-sustaining distributed storage system must  possess the ability to  repair the failed node. For example, to accomplish this task,
 the repair process of the classical MDS codes is  to first reconstruct the
original file by accessing and downloading an amount of data equal to the size of the original file, and then repair the failed node, which is called  \textit{naive repair}.
However,  such repair is overly excessive, and poses the question whether we can minimize the \textit{rebuilding access} and the \textit{repair bandwidth},
i.e., the amount of data that needs to be accessed (or read) and downloaded to repair a failed node, respectively.

The seminal work in \cite{Dimakis} gave a positive answer to the above question. 
It was shown in \cite{Dimakis} that the minimum repair bandwidth of an $[n, k]$ MDS code $\mathcal{C}$ defined over $\mathbf{F}_q^{\alpha}$ (which is referred to as an \textit{array code} \cite{array_codes} if $\alpha>1$) is $\gamma^*(d)\triangleq\frac{d}{d-k+1}\alpha$, where $d\in [k, n)$ stands for the number of helper nodes contacted  during the repair procedure. MDS codes attaining the minimum repair bandwidth are said to have \textit{optimal repair bandwidth} and are referred to as  minimum storage regenerating (MSR) codes.  Subsequently, the optimal rebuilding access was also established  in \cite{Barg1,tight-bound-on-alpha},
where an MDS code is said to possess \textit{optimal rebuilding access}  if the amount of data accessed meets the minimum repair bandwidth during the repair process of a failed node. 
Since 2010,  various MDS array codes  with optimal repair bandwidth
have been proposed in the literature \cite{Goparaju,product,Suh,hadamard,Hadamard-strategy,Sasidharan-Kumar2,extend-zigzag,Zigzag,repair-parity-zigzag,invariant-subspace,Etzion,Long-IT,MDR,MDR-M,EEGad,Barg1,Barg2,elyasi2019cascade,han2015update}, where most works \cite{Suh,hadamard,Hadamard-strategy,extend-zigzag,Zigzag,repair-parity-zigzag,invariant-subspace,Etzion,Long-IT,MDR,MDR-M,EEGad} consider the case $d=n-1$ to maximally reduce the repair bandwidth, as $\gamma^*(d)$ is a decreasing function of $d$;  the setting $d=n-1$ is also the focus of this work, then the optimal repair can be achieved by downloading $\frac{\alpha}{n-k}$ elements from each surviving nodes as
\begin{equation}\label{Eqn-optRB}
\gamma^*(n-1)=\frac{n-1}{n-k}\alpha.
\end{equation}

It is worth noting that a symbol in $\mathbf{F}_q^b$ can be regarded as a symbol in $\mathbf{F}_{q^b}$, and vice versa. Thus clearly, every linear code over $\mathbf{F}_{q^b}$ can be regarded as an $\mathbf{F}_q$-linear code over  $\mathbf{F}_q^b$, but the converse is not true \cite{blaum1999lowest}. RS codes over $\mathbf{F}_{q^b}$ are examples of $\mathbf{F}_q$-linear MDS array codes over $\mathbf{F}_q^b$. Repairing RS codes has also been attracting a lot of attention recently, however, in general, they either cannot attain optimal repair bandwidth or should be constructed over a large finite field to achieve optimal repair bandwidth.  Interested readers could refer to \cite{guruswami2017repairing,chen2020enabling} and the references therein for more details.


For high-rate (i.e., $k/n > 1/2$) MDS array codes, most early systematic code constructions can only optimally repair all the systematic nodes but not all the parity nodes \cite{hadamard,Hadamard-strategy,Zigzag,invariant-subspace,Etzion,Long-IT,MDR,MDR-M,EEGad
}. Fortunately, this issue was  first solved by Li \textit{et al.} \cite{transform-IT,transform-ISIT17,transform-ISIT18}, where  a   transformation that can enable optimal node repair for any $n-k$ chosen nodes of a  nonbinary $[n, k]$   MDS  code was proposed, and at the same time, preserves the repair efficiency
for the remaining $k$ nodes. 
It is worth noting that
independent and parallel to the works in \cite{transform-IT,transform-ISIT17,transform-ISIT18}, Ye \textit{et al.}  \cite{Barg2}  and  Sasidharan \textit{et al.} \cite{Sasidharan-Kumar2}  respectively proposed  explicit constructions of high-rate MDS array codes, which are equivalent in essence, and have the same performance as the ones obtained from the second application in \cite{transform-IT,transform-ISIT17,transform-ISIT18} (see  \cite{tight-bound-on-alpha} for the discussion on these discoveries).

The generic transformation   proposed  in \cite{transform-IT,transform-ISIT17,transform-ISIT18} has wide potential applications,  however,  it  does not work for binary MDS array codes.
In this paper, we aim to address the unsolved problem in \cite{transform-IT} by presenting a new generic transformation that can convert any  binary   MDS array code into another one, which endows any $r=n-k\ge 2$ chosen nodes with optimal rebuilding access, and at the same time, preserves the repair efficiency 
for the remaining $k$ nodes under some conditions.
As two
immediate applications  of this transformation, we show that
1) any   binary MDS array code can be converted into one with optimal rebuilding access for all nodes, 2) any binary   MDS array code with optimal repair bandwidth
or optimal rebuilding access for the systematic nodes  only can
be converted into one with the corresponding 
optimality property for all nodes. In fact, the second  application  can also be applied to any binary MDS array code with an efficient but not optimal repair   for the systematic nodes, such as the MDS array codes in  \cite{EVENODD,hou20120two},  to   enable optimal repair of all the parity nodes while keeping the repair efficiency of the systematic nodes.    The proposed transformation is also applicable to nonbinary MDS array codes and RS codes, but the performance is no better than those in \cite{transform-IT,transform-ISIT17,transform-ISIT18}.

We note that   another transformation  for binary MDS array codes was proposed recently  in \cite{HouLee}, which has a similar flavor as the one proposed in this paper.
The main technique used  in the transformation  in \cite{HouLee} is similar to that in \cite{transform-IT}, i.e., by operating on multiple instances of a base code and pairing data. 
The difference is that  in \cite{HouLee}  the data stored in each node is regarded as a polynomial of degree less than the sub-packetization level
and thus   the pairing step is carried out in a polynomial ring.
Compared with the one in \cite{HouLee}, the generic transformation proposed in this paper has several advantages, such as providing a uniform procedure, having a wide potential for applications, and enabling low complexity. Detailed comparisons will be carried out  in Section \ref{sec:comp}. Very recently, as a variant of the transformation in \cite{HouLee}, another transformation that was tailored for  Zigzag-decodable reconstruction (ZDR) codes was proposed in \cite{hou2020zigzag}. ZDR codes are not MDS, hence they are left out of consideration in this paper. Although it was mentioned that the recent transformation in \cite{hou2020zigzag} can also be applied to binary MDS array codes, however, it is only applicable to those whose repair matrices satisfy a strict condition (cf. \cite[Theorem 7]{hou2020zigzag}), which greatly limits its application and makes it inferior to the transformation in \cite{hou2020zigzag} in terms of the range of the application to binary MDS array codes.




The remainder of the paper is organized as follows. Section
II gives some necessary preliminaries. The generic transformation is given in Section III, followed by the proofs of
the asserted properties. Two specific  applications    of the transformation as well as comparisons  are discussed in Section IV. Finally, Section V draws the conclusion.


%

\section{Preliminaries}

An $[N, K, D]$ \textit{(linear) code} over $\mathbf{F}_q$ is   a subspace of $\mathbf{F}_q^N$, where the minimum Hamming distance $D$ is the smallest weight of the nonzero codewords or nonzero vectors in the subspace. The  Singleton bound implies that $D$ satisfies $D\le N-K+1$ for any linear code, linear codes with $D= N-K+1$ are called MDS codes and denoted as $[N, K]$ MDS codes as the minimum distance is clear \cite{ling2004coding}.
An $[n, k]$ array code $\mathcal{C}$ over $\mathbf{F}_q^b$ is said to be $\mathbf{F}_q$\textit{-linear} if $\mathcal{C}$ is a vector space over $\mathbf{F}_q$ of dimension $nb$, and is MDS if the Hamming distance equals $n-k+1$, where the Hamming distance  is measured respect to the symbols in $\mathbf{F}_q^b$.  
An array code can be specified by  either its parity-check matrix $H$ of size $(n-k)b\times nb$ or its generator matrix $G$ of size $kb\times nb$, both over $\mathbf{F}_q$. The code is \textit{systematic} if $H$ contains $\mathbf{I}_{(n-k)b}$ or $G$ contains   $\mathbf{I}_{kb}$, where
$\mathbf{I}_{m}$   denotes the identity matrix of order  $m$ \cite{blaum1999lowest}.

In this paper, we only focus on binary MDS array codes, i.e.,  $\mathbf{F}_2$-linear MDS code  over $\mathbf{F}_2^{\alpha}$ for some integer $\alpha$, which is called the \textit{sub-packetization level}, thus all the scalars, vectors, and matrices are binary.
First of all, we fix  some notations, which will  be used throughout this paper.
\begin{itemize}
  \item  ``$+$" is just the XOR operation in ``$a+b$" for $a,b\in \mathbf{F}_2$.
  \item  A lowercase letter  in bold (for example, $\mathbf{a}$)  denotes a binary column vector.
\item   $[i,j)$  denotes the set  $\{i,i+1,\cdots,j-1\}$ for any two integers $i<j$.
\end{itemize}

Assume the original file to be stored is comprised of $M=k\alpha$ symbols over $\mathbf{F}_2$. Encoding the original file by an  $[n=k+r,k]$  MDS array code over $\mathbf{F}_2^{\alpha}$, $n$ column vectors $\mathbf{f}_0, \mathbf{f}_1, \cdots, \mathbf{f}_{n-1}$ of length $\alpha$ are obtained. As common practice (including in the seminal work \cite{Dimakis}), we assume that the data in these $n$ column vectors is dispersed
across $n$ storage nodes, with each node storing one column vector.
If the code is systematic, then the  $k$ nodes storing the original file are named the \textit{systematic nodes} while the remaining nodes
are referred to as the \textit{parity nodes}. 

An $[n,k]$ MDS array code possesses \textit{the MDS property} that the original file can be reconstructed by contacting any $k$ out of the $n$ nodes, and is preferable  to have optimal repair bandwidth, i.e.,  a failed node can be regenerated by downloading $\alpha/r$ symbols from each surviving node \cite{Dimakis}. Generally, the data downloaded from node $j$ to repair node $i$ can be represented by $S_{i,j}\mathbf{f}_j$, where $S_{i,j}$ is an $\alpha\times \alpha$ matrix with its rank indicating the amount of data that should be downloaded,  and $S_{i,j}$ is usually referred to as \textit{repair matrix}. Besides the repair bandwidth, the rebuilding access (also known as repair access in \cite{HouLee}) should also be optimized. Formally, rebuilding access is the amount of data that needs to be accessed (or read) to repair a failed node, which is of course no less than the repair bandwidth. It would be preferable for an MDS array code to have  optimal rebuilding access, which can be achieved if only $\alpha/r$ symbols are accessed at each surviving node \cite{Barg2}, i.e., there are exactly $\alpha/r$ nonzero columns in $S_{i,j}$, where $i,j\in [0, n)$ with $i\ne j$. This appealing property reinforces the repair bandwidth requirement, and can reduce the disk I/O overhead during the repair process.

\begin{table*}[htbp]
\begin{center}
\caption{A $[9,6]$  MDS array  code with optimal rebuilding access for the target nodes}\label{Table stru new}
\setlength{\tabcolsep}{5.4pt}
\begin{tabular}{|c|c|c|l|l|l|}
\hline  RN 0 ($\mathbf{f}_0$)  & $\cdots$ & RN $5$ ($\mathbf{f}_{5}$)  & TN $0$ ($\mathbf{h}_{0}$) &TN $1$ ($\mathbf{h}_{1}$) & TN $2$ ($\mathbf{h}_{2}$) \\
\hline
\hline \multirow{2}{*}{$\mathbf{f}_0^{(0)}$} & \multirow{2}{*}{$\cdots$} & \multirow{2}{*}{$\mathbf{f}_{5}^{(0)}$} & $\mathbf{g}_{0,0}^{(0)}$ & $\mathbf{g}_{1,0}^{(0)}+\mathbf{g}_{0,0}^{(1)}$ & $\mathbf{g}_{2,0}^{(0)}+\mathbf{g}_{0,0}^{(2)}$\\
&&&$\mathbf{g}_{0,1}^{(0)}$& $\mathbf{g}_{1,1}^{(0)}+\mathbf{g}_{0,1}^{(1)}$ & $\mathbf{g}_{2,1}^{(0)}+\mathbf{g}_{0,1}^{(2)}$\\
\hline  \multirow{2}{*}{$\mathbf{f}_0^{(1)}$} & \multirow{2}{*}{$\cdots$} & \multirow{2}{*}{$\mathbf{f}_{5}^{(1)}$} &  $\mathbf{g}_{0,0}^{(1)}+\mathbf{g}_{1,0}^{(0)}+\mathbf{g}_{1,1}^{(0)}$ & $\mathbf{g}_{1,0}^{(1)}$ &  $\mathbf{g}_{2,0}^{(1)}+\mathbf{g}_{1,0}^{(2)}$\\
&&&$\mathbf{g}_{0,1}^{(1)}+\mathbf{g}_{1,0}^{(0)}$ & $\mathbf{g}_{1,1}^{(1)}$ &  $\mathbf{g}_{2,1}^{(1)}+\mathbf{g}_{1,1}^{(2)}$\\
\hline \multirow{2}{*}{$\mathbf{f}_0^{(2)}$} & \multirow{2}{*}{$\cdots$} & \multirow{2}{*}{$\mathbf{f}_{5}^{(2)}$} & $\mathbf{g}_{0,0}^{(2)}+\mathbf{g}_{2,0}^{(0)}+\mathbf{g}_{2,1}^{(0)}$ & $\mathbf{g}_{1,0}^{(2)}+\mathbf{g}_{2,0}^{(1)}+\mathbf{g}_{2,1}^{(1)}$ &  $\mathbf{g}_{2,0}^{(2)}$\\
&&&$\mathbf{g}_{0,1}^{(2)}+\mathbf{g}_{2,0}^{(0)}$ & $\mathbf{g}_{1,1}^{(2)}+\mathbf{g}_{2,0}^{(1)}$ &  $\mathbf{g}_{2,1}^{(2)}$\\
\hline
\end{tabular}
\end{center}
\end{table*}

\section{A Generic Transformation for Binary MDS Array Codes}\label{Sec tran}

In this section,  we present a generic transformation that can convert any $[n,k]$ binary  MDS array code with even sub-packetization level\footnote{Note that most known binary MDS array codes in the literature have even sub-packetization level  \cite{EVENODD,Blaum_Bruck_Vardy96,EVENODDG,STAR,HouLee,EEGad}. Even when the sub-packetization level of a binary MDS array code is odd,  one can combine two instances of such a code in advance so that the sub-packetization level of the resultant code is even.} into a new one, which allows an arbitrary set of $r=n-k$ nodes to have both optimal repair bandwidth and optimal rebuilding access. Choosing some $[n,k]$ binary MDS array code as the base code,  the $r$ nodes which we wish to endow with optimal rebuilding access are named the \textit{target nodes}, while the remaining  $k$ nodes are called the \textit{remainder nodes}. W.L.O.G., we always assume  that the last $r$ nodes are the target nodes unless otherwise stated. To simplify the notation in what follows,  a target node (resp. remainder node) is shortly denoted by TN (resp. RN).  Before describing the generic transformation,  we first present an example  to illustrate the main idea behind it.

\subsection{An Example $[9,6]$ MDS array Code}
\label{ex_step2}
Let $\mathcal{C}_{1}$ be a  $[9,6]$ MDS array code  over   $\mathbf{F}_2^{\alpha'}$, where $\alpha'$ is even.  We generate  three codewords of the given code. To this end, let $(\mathbf{f}_{0}^{(l)},\mathbf{f}_{1}^{(l)},\cdots,\mathbf{f}_{5}^{(l)},\mathbf{g}_{0}^{(l)},\mathbf{g}_{1}^{(l)},\mathbf{g}_{2}^{(l)})$ be a codeword/instance of the base code $\mathcal{C}_{1}$, where $l\in [0,3)$ and the nodes storing $\mathbf{g}_{0}^{(l)},\mathbf{g}_{1}^{(l)},\mathbf{g}_{2}^{(l)}$ are designated as the target nodes. That is, the original file is first encoded into three codewords of the base code. For each code symbol, we divide it into two equal parts. That is, rewrite  $\mathbf{f}_{i}^{(l)}$ ($i\in [0,5)$) and $\mathbf{g}_{j}^{(l)}$ ($j\in [0,3)$) as
\begin{eqnarray*}
\mathbf{f}_i^{(l)}=\left(
                                                                                                                            \begin{array}{c}
                                                                                                                              \mathbf{f}_{i,0}^{(l)} \\
                                                                                                                              \mathbf{f}_{i,1}^{(l)} \\
                                                                                                                            \end{array}
                                                                                                                          \right)
,~\mathbf{g}_j^{(l)}=\left(
                     \begin{array}{c}
                       \mathbf{g}_{j,0}^{(l)} \\
                       \mathbf{g}_{j,1}^{(l)} \\
                     \end{array}
                   \right).
\end{eqnarray*}
By applying the generic transformation,  a new $[9,6]$ MDS array code with sub-packetization level $\alpha=3\alpha'$ is obtained, which is shown in Table \ref{Table stru new}. The new code has  optimal rebuilding access for the $3$ target nodes,  which are denoted as

\begin{equation*}
\mathbf{h}_{i}=\left(
                   \begin{array}{c}
                     \mathbf{h}_i^{(0)} \\
                     \mathbf{h}_i^{(1)} \\
                     \mathbf{h}_i^{(2)} \\
                   \end{array}
                 \right)~~\mbox{and}~~\mathbf{h}_i^{(l)}=\left(
                                                            \begin{array}{c}
                                                              \mathbf{h}_{i,0}^{(l)} \\
                                                              \mathbf{h}_{i,1}^{(l)} \\
                                                            \end{array}
                                                          \right),
                            \end{equation*}             where
                           \begin{equation*}
 \mathbf{h}_{i,j}^{(l)}=\left\{ \begin{array}{ll}
\mathbf{g}_{i,j}^{(i)}, & \textrm{if}~i=l, \\
\mathbf{g}_{i,0}^{(l)}+\mathbf{g}_{l,0}^{(i)}+\mathbf{g}_{l,1}^{(i)},& \textrm{if}~i<l ~\textrm{and}~j=0,\\
\mathbf{g}_{i,1}^{(l)}+\mathbf{g}_{l,0}^{(i)},& \textrm{if}~i<l ~\textrm{and}~j=1,\\
\mathbf{g}_{i,j}^{(l)}+\mathbf{g}_{l,j}^{(i)},& \textrm{if}~i>l,
\end{array}\right. 
\end{equation*}
for $i,l\in[0,3)$.

\textbf{Reconstruction:} W.L.O.G., let us consider reconstructing  the original file by contacting nodes $2$ to $7$ (i.e., $\mathbf{f}_2, \ldots \mathbf{f}_5$, $\mathbf{h}_{0}$, and $\mathbf{h}_{1}$).
According to Table \ref{Table stru new}, we can get the components of $\mathbf{g}_{0}^{(1)}$ and $\mathbf{g}_{1}^{(0)}$ in the following manner.
\begin{eqnarray*}
  &&\mathbf{g}_{1,1}^{(0)}= \mathbf{h}_{0,0}^{(1)}+ \mathbf{h}_{1,0}^{(0)}, ~
  \mathbf{g}_{0,1}^{(1)}= \mathbf{g}_{1,1}^{(0)}+\mathbf{h}_{1,1}^{(0)},\\ 
  &&
  \mathbf{g}_{1,0}^{(0)}=  \mathbf{g}_{0,1}^{(1)}+\mathbf{h}_{0,1}^{(1)}, ~
  \mathbf{g}_{0,0}^{(1)}=\mathbf{g}_{1,0}^{(0)}+\mathbf{h}_{1,0}^{(0)}.
\end{eqnarray*}
In conjunction with the remaining data in  rows 1 and 2 at nodes $2$ to $7$, we now obtain
\begin{align*}
(\mathbf{f}_2^{(0)},\ldots,\mathbf{f}_5^{(0)}, \mathbf{g}_0^{(0)}, \mathbf{g}_1^{(0)}),~
(\mathbf{f}_2^{(1)},\ldots,\mathbf{f}_5^{(1)}, \mathbf{g}_0^{(1)}, \mathbf{g}_1^{(1)}),
\end{align*}
from which $(\mathbf{f}_0^{(0)},\ldots,\mathbf{f}_5^{(0)})$ and $(\mathbf{f}_0^{(1)},\ldots,\mathbf{f}_5^{(1)})$ can be reconstructed, respectively, according to the MDS property of the base code.
Then $\mathbf{g}_2^{(0)}$ and $\mathbf{g}_2^{(1)}$ can   be computed from these available data,  and together with $\mathbf{h}_0^{(2)}$ and $\mathbf{h}_1^{(2)}$, we can obtain $\mathbf{g}_0^{(2)}$ and $\mathbf{g}_1^{(2)}$ since
\begin{equation*}
  \mathbf{g}_0^{(2)}=\left(
                       \begin{array}{c}
                         \mathbf{g}_{0,0}^{(2)} \\
                         \mathbf{g}_{0,1}^{(2)} \\
                       \end{array}
                     \right)=\left(
                       \begin{array}{l}
                         \mathbf{h}_{0,0}^{(2)}+\mathbf{g}_{2,0}^{(0)}+\mathbf{g}_{2,1}^{(0)} \\
                         \mathbf{h}_{0,1}^{(2)}+ \mathbf{g}_{2,0}^{(0)}\\
                       \end{array}
                     \right),
\end{equation*}
\begin{equation*}
\mathbf{g}_1^{(2)}=\left(
                       \begin{array}{c}
                         \mathbf{g}_{1,0}^{(2)} \\
                         \mathbf{g}_{1,1}^{(2)} \\
                       \end{array}
                     \right)=\left(
                       \begin{array}{l}
                         \mathbf{h}_{1,0}^{(2)}+\mathbf{g}_{2,0}^{(1)}+\mathbf{g}_{2,1}^{(1)} \\
                         \mathbf{h}_{1,1}^{(2)}+ \mathbf{g}_{2,0}^{(1)}\\
                       \end{array}
                     \right).
\end{equation*}

Finally, in conjunction  with the remaining data in the last row at nodes $2$ to $5$,
we now obtain
\begin{align*}
&(\mathbf{f}_2^{(2)},\ldots,\mathbf{f}_5^{(2)}, \mathbf{g}_0^{(2)}, \mathbf{g}_2^{(2)}),
\end{align*}
from which  $(\mathbf{f}_0^{(2)},\ldots,\mathbf{f}_5^{(2)})$ can be reconstructed. The original file is therefore  reconstructed according to the above analysis.

\textbf{Optimal rebuilding access for the target nodes:} W.L.O.G., consider repairing target node $0$, and we download   the following data
\begin{align*}
\mathbf{f}_{0}^{(0)}, \ldots, \mathbf{f}_{5}^{(0)}, \mathbf{h}_{1}^{(0)},\mathbf{h}_{2}^{(0)},
\end{align*}
i.e., the data in row 1 of Table \ref{Table stru new}.
Clearly, $\mathbf{g}_0^{(0)}$ can be obtained from $\mathbf{f}_{0}^{(0)}, \ldots, \mathbf{f}_{5}^{(0)}$. To compute
$$\mathbf{h}_{0}^{(1)}=\left(
                        \begin{array}{l}
                     \mathbf{g}_{0,0}^{(1)}+\mathbf{g}_{1,0}^{(0)}+\mathbf{g}_{1,1}^{(0)} \\
                          \mathbf{g}_{0,1}^{(1)}+\mathbf{g}_{1,0}^{(0)}  
                        \end{array}
                      \right)$$ that was
 stored at target node $0$, observe firstly that $\mathbf{g}_{1}^{(0)}=\left(
                                     \begin{array}{c}
                                        \mathbf{g}_{1,0}^{(0)}\\
                                       \mathbf{g}_{1,1}^{(0)} \\
                                     \end{array}
                                   \right)
$ can be computed with $\mathbf{f}_{0}^{(0)},\mathbf{f}_{1}^{(0)},\ldots,\mathbf{f}_{5}^{(0)}$,  then we can regenerate $\mathbf{g}_0^{(1)}=\left(
                                     \begin{array}{c}
                                        \mathbf{g}_{0,0}^{(1)}\\
                                       \mathbf{g}_{0,1}^{(1)} \\
                                     \end{array}
                                   \right)$  from the downloaded data $\mathbf{h}_{1}^{(0)}=\mathbf{g}_{1}^{(0)}+\mathbf{g}_0^{(1)}$, and obtain $\mathbf{h}_0^{(1)}$ subsequently. The other piece of coded data $\mathbf{h}_0^{(2)}$ stored at target node $0$ can be similarly regenerated. Thus  target node $0$ can indeed be optimally repaired  and has optimal rebuilding access  according to \eqref{Eqn-optRB}.

\subsection{A Key Pairing}

In this subsection,  we introduce a pairing technique of two column vectors and   analyze its properties, which will be crucial for the generic transformation.

Let $N$  denote an even constant from now on. For any two column vectors   $\mathbf{a}{[i]}$
and $\mathbf{b}{[i]}$ of length $N$,  we  divide them into two equal parts, which can be represented  as
\begin{eqnarray*}
 \mathbf{a}[i]=\left(
                 \begin{array}{c}
                   \mathbf{a}[i,0] \\
                  \mathbf{a}[i,1] \\
                 \end{array}
               \right),~~
   \mathbf{b}[i]=\left(
                 \begin{array}{c}
                   \mathbf{b}[i,0] \\
                  \mathbf{b}[i,1] \\
                 \end{array}
               \right).
\end{eqnarray*}
Then we define a linear operation $\boxplus$ between two column vectors $\mathbf{a}{[i]}$ and $\mathbf{b}{[i]}$  of length $N$  as
\begin{eqnarray}\label{Eqn oplus}
 \nonumber \mathbf{a}[i]\boxplus \mathbf{b}[i]
            & =&\mathbf{a} [i]+\left(
                 \begin{array}{cc}
                  \mathbf{I}_{N/2} & \mathbf{I}_{N/2} \\
                  \mathbf{I}_{N/2}&\mathbf{0}_{N/2}\\
                 \end{array}
               \right)\mathbf{b} [i]\\&=&\left(
                 \begin{array}{l}
 \mathbf{a}[i,0]+\mathbf{b}[i,0]+\mathbf{b}[i,1] \\
                  \mathbf{a}[i,1]+\mathbf{b}[i,0] \\
                 \end{array}
               \right),
\end{eqnarray}
which takes $\frac{3N}{2}$ XORs, where  $\mathbf{0}_{N/2}$  denotes the  zero matrix of order $N/2$.

For any two column vectors $\mathbf{a}$ and $\mathbf{b}$ of length $\delta N$, where $\delta\ge 1$, we divide them into $\delta$ segments,  i.e.,  rewrite $\mathbf{a}$ and $\mathbf{b}$ as
\begin{eqnarray}\label{Eqn divide a}
\mathbf{a}=\left(
    \begin{array}{c}
      \mathbf{a}{[0]} \\
      \mathbf{a}{[1]} \\
      \vdots \\
      \mathbf{a}{[\delta-1]} \\
    \end{array}
  \right),~~\mathbf{b}=\left(
    \begin{array}{c}
      \mathbf{b}{[0]} \\
      \mathbf{b}{[1]} \\
      \vdots \\
      \mathbf{b}{[\delta-1]} \\
    \end{array}
  \right),
\end{eqnarray}
where $\mathbf{a}{[i]}$ and $\mathbf{b}{[i]}$ are column vectors of length $N$ and are named  the $i$th \textit{segments} of $\mathbf{a}$ and $\mathbf{b}$ for $i\in [0,\delta)$, respectively.
Based on \eqref{Eqn oplus}, we further define a linear operation   $\boxplus_N$ between  two column vectors  $\mathbf{a}$ and $\mathbf{b}$ of length $\delta N$  as
\begin{eqnarray}\label{Eqn Oplus}
 \mathbf{a}\boxplus_N \mathbf{b}=\left(
    \begin{array}{c}
      \mathbf{a}{[0]}\boxplus \mathbf{b}{[0]} \\
      \mathbf{a}{[1]}\boxplus \mathbf{b}{[1]} \\
      \vdots \\
      \mathbf{a}{[\delta-1]}\boxplus \mathbf{b}{[\delta-1]} \\
    \end{array}
  \right),
\end{eqnarray}
i.e., performing  the linear operation $\boxplus$ defined  in \eqref{Eqn oplus} on the each segments of $\mathbf{a}$ and $\mathbf{b}$, which takes $\frac{3\delta N}{2}$ XORs.
Then,  the  following fact is obvious.
\begin{Fact}\label{Fac Pairing G}
For any two column vectors $\mathbf{a}$ and $\mathbf{b}$ of length $\delta N$ with $\delta \ge 1$, one can get
\begin{itemize}
  \item [(i)]  $\mathbf{a}[i]$ and $\mathbf{b}[i]$ from $\mathbf{a}[i]+\mathbf{b}[i]$ and $\mathbf{a}[i]\boxplus_N \mathbf{b}[i]$ by $2N$ XORs,
  \item [(ii)] $\mathbf{a}[i]$  from $\mathbf{b}[i]$  and $\mathbf{a}[i]\boxplus_N \mathbf{b}[i]$ (or $\mathbf{b}[i]\boxplus_N \mathbf{a}[i]$ or $\mathbf{a}[i]+ \mathbf{b}[i]$) by at most $\frac{3N}{2}$ XORs,
  \item [(iii)]  $S\cdot \mathbf{a}[i]$ and $S\cdot\mathbf{b}[i]$ from $S(\mathbf{a}[i]+\mathbf{b}[i])$ and $S(\mathbf{a}[i]\boxplus_N \mathbf{b}[i])$ by $4L$ XORs,
\end{itemize}
for all $i\in [0,\delta)$  and $2L\times N$ matrix
$S=\left(\begin{array}{cc}
S' & \\
& S'
\end{array}
\right)$ with $L$ being any positive integer.
\end{Fact}

\begin{Remark}\label{Remark-interprete}
The operation in \eqref{Eqn oplus} can be viewed as multiple basic linear operations between two binary column vectors of length $2$. This basic linear operation can    also be interpreted as  an operation  between two elements in $\mathbf{F}_{2^2}$ or two polynomials in $\mathbf{F}_{2}[x]$ with degree less than $2$. Such equivalent interpretations also hold for the operation in \eqref{Eqn Oplus}.
\end{Remark}

\subsection{A Generic Transformation}\label{subsec generic trans}
In this subsection, we propose the generic transformation, which employs a known $[n=k+r,k]$ binary  MDS array code $\mathcal{C}_{1}$  with  sub-packetization level $\alpha'=\delta N$ as the base code, where  $\delta \ge 1$ and  $r\ge 2$.
The transformation is then carried  out through the following three steps.

\subsection*{\textbf{Step 1: GENERATING $r$ instances of the base code $\mathcal{C}_{1}$ }}
Assume generating $r$ instances of the   code $\mathcal{C}_{1}$ to obtain an intermediate MDS array code $\mathcal{C}_2$ with sub-packetization level $\alpha=r\alpha'$, and 
denote by
$\mathbf{f}_{i}^{(l)}
$ and $\mathbf{g}_{j}^{(l)}
$  the data stored at remainder node $i$ and target node $j$ of the $l$-th instance of  $\mathcal{C}_{1}$, respectively,   where $i\in [0,k)$ and $l,j\in [0,r)$.

\subsection*{\textbf{Step 2: PERMUTING the data in the target nodes of $\mathcal{C}_2$}}

Keeping the  remainder
nodes of $\mathcal{C}_2$ intact, we construct another intermediate MDS array code $\mathcal{C}_3$ by permuting the data in the target nodes of $\mathcal{C}_2$, which is illustrated as follows. Denote by  $\mathbf{h}_{j}$   the data stored at target node $j$ of code $\mathcal{C}_3$. Let us write $\mathbf{h}_{j}$ as \begin{equation*}
  \mathbf{h}_{j} = \left(
                         \begin{array}{c}
                           \mathbf{h}_{j}^{(0)} \\
                           \vdots \\
                           \mathbf{h}_{j}^{(r-1)} \\
                         \end{array}
                       \right), j\in [0,r)
\end{equation*}
for convenience,
where $\mathbf{h}_{j}^{(l)}$ is defined as
\begin{eqnarray}\label{Eqn perm h}
\mathbf{h}_{j}^{(l)}=\mathbf{g}_{\pi_l(j)}^{(l)}, ~~j,l\in [0,r),
\end{eqnarray}
and $\pi_0,\pi_1,\cdots,\pi_{r-1}$ are some $r$ specific permutations on $[0,r)$, more details of the  requirements of the permutations   will be given  in Theorem \ref{Thm repair sys}.

\begin{table*}[htbp]
\begin{center}
\caption{The new  code  $\mathcal{C}_4$}\label{Table C3}
\begin{tabular}{|c|c|c|c|c|c|c|c|}
\hline   RN 0 & $\cdots$  & RN $k-1$ & TN $0$ ($\mathbf{h'}_{0}$)& TN $1$ ($\mathbf{h'}_{1}$)& $\cdots$  &TN $r-1$ ($\mathbf{h'}_{r-1}$)\\
\hline\hline $\mathbf{f}_0^{(0)}$ &  $\cdots$ & $\mathbf{f}_{k-1}^{(0)}$ &$\mathbf{h}_0^{(0)}$ &  $\mathbf{h}_1^{(0)}+\mathbf{h}_0^{(1)}$ &
$\cdots$ & $\mathbf{h}_{r-1}^{(0)}+\mathbf{h}_0^{(r-1)}$ \\
\hline  $\mathbf{f}_0^{(1)}$ &  $\cdots$  &  $\mathbf{f}_{k-1}^{(1)}$& $\mathbf{h}_0^{(1)}\boxplus_N\mathbf{h}_1^{(0)}$&$\mathbf{h}_1^{(1)}$&  $\cdots$  & $\mathbf{h}_{r-1}^{(1)}+\mathbf{h}_1^{(r-1)}$ \\
\hline  $\vdots$&$\ddots$ & $\vdots$&$\vdots$&$\vdots$&$\ddots$ &$\vdots$\\
\hline
  $\mathbf{f}_0^{(r-1)}$ & $\cdots$ & $\mathbf{f}_{k-1}^{(r-1)}$ & $\mathbf{h}_0^{(r-1)}\boxplus_N\mathbf{h}_{r-1}^{(0)}$& $\mathbf{h}_1^{(r-1)}\boxplus_N\mathbf{h}_{r-1}^{(1)}$ & $\cdots$ & $\mathbf{h}_{r-1}^{(r-1)}$ \\
\hline
\end{tabular}
\end{center}
\end{table*}

\subsection*{\textbf{Step 3: PAIRING the data in the target nodes of $\mathcal{C}_3$}}

By modifying only the data at the target nodes of the code $\mathcal{C}_3$ while keeping its remainder nodes   intact,
we construct the desired storage code $\mathcal{C}_4$ as follows. Denote by $\mathbf{h'}_{j}$ the data stored at target node $j$  of code $\mathcal{C}_4$. Let us write $\mathbf{h'}_{j}$ as
\begin{equation*}
  \mathbf{h'}_{j} =\left(
                         \begin{array}{c}
                           \mathbf{h'}_{j}^{(0)} \\
                           \vdots\\
                           \mathbf{h'}_{j}^{(r-1)} \\
                         \end{array}
                       \right)
\end{equation*}
for convenience,
where $\mathbf{h'}_{j}^{(l)}$,  $j,l\in[0,r)$,
are
defined by
\begin{eqnarray}\label{Eqn new con}
\mathbf{h'}_{j}^{(l)}=\left\{ \begin{array}{ll}
\mathbf{h}_{j}^{(j)}, & \textrm{if}~j=l, \\
\mathbf{h}_{j}^{(l)}+\mathbf{h}_{l}^{(j)},& \textrm{if}~j>l,\\
\mathbf{h}_{j}^{(l)}\boxplus_N \mathbf{h}_{l}^{(j)}
, &  \textrm{if}~j<l,
\end{array}\right.
\end{eqnarray}
and the linear operation  $\boxplus_N$ is defined  in \eqref{Eqn Oplus}. It is easy to see that this step takes a total of $\frac{5r(r-1)}{4}\alpha'$ XORs.
The  new code $\mathcal{C}_4$ is  depicted in Table  \ref{Table C3}.


In the following, we show that the new code $\mathcal{C}_4$ maintains the MDS property of the base code. The proof is similar to that in \cite{transform-IT}, nevertheless, we include it for completeness.

\begin{Theorem}\label{Thm MDS C4}
Code $\mathcal{C}_4$ has the MDS property. More specifically, the reconstruction process of the code $\mathcal{C}_4$ requires invoking the reconstruction process of the base code $r$ times and at most $\frac{1}{2}(3rt-t^2-2t)\alpha'$ additional XORs, where $t$ denotes the number of target nodes that are connected during the reconstruction, for $0\le t\le \min\{r,k\}$.
\end{Theorem}
\begin{proof}
The  code $\mathcal{C}_4$ has the MDS property if the original file (or equivalently, $\mathbf{f}_{i}^{(l)}$, $i\in [0,k)$ and $l\in [0,r)$) can be reconstructed by connecting to any $k$ out of the $n$ nodes.
The reconstruction is discussed in the following two cases.

(i) When connecting  to   $k$ remainder nodes: the original file can be reconstructed according to the MDS property of the base code.

(ii) Suppose $k-t$ remainder nodes and $t$ target nodes are connected, where $1\le t\le \min\{r,k\}$:   let $I=\{i_0,i_1,\cdots,i_{t-1}\}$  denote  the set of the indices of  the $t$ remainder nodes that are not connected and $J=\{j_0,j_1,\cdots,j_{t-1}\}$ be the set of the indices of the connected target nodes, where $0\le i_0<\cdots<i_{t-1}<k$, $0\le j_0<\cdots<j_{t-1}<r$. Denote  $\{j_t,\cdots,j_{r-1}\}=[0,r)\backslash J$. For example, in the example in Section \ref{ex_step2}, $t=2$, $I=\{0,1\}$, $J=\{0,1\}$, and $j_2=2$.

Firstly,    from the $t$ connected target nodes, the data in Table \ref{Table Thm MDS 1} is available, from which in conjunction with Fact \ref{Fac Pairing G}-(i),  we can get   $\mathbf{h}_{u}^{(l)}$ ($l,u\in J$) by $t(t-1)\alpha'$  XORs.

\begin{table}[htbp]
\begin{center}
\caption{Data from the $t$ connected target nodes
}\label{Table Thm MDS 1}
\begin{tabular}{|c|c|c|c|c|c|c|c|c|}
\hline   TN $j_0$ & TN $j_1$ & $\cdots$  &TN $j_{t-1}$ \\
\hline
\hline $\mathbf{h}_{j_0}^{(j_0)}$& $\mathbf{h}_{j_1}^{(j_0)}+\mathbf{h}_{j_0}^{(j_1)}$ &$\cdots$ & $\mathbf{h}_{j_{t-1}}^{(j_0)}+\mathbf{h}_{j_0}^{(j_{t-1})}$\\
\hline $\mathbf{h}_{j_0}^{(j_1)}\boxplus_N \mathbf{h}_{j_1}^{(j_0)}$& $\mathbf{h}_{j_1}^{(j_1)}$& $\cdots$ & $\mathbf{h}_{j_{t-1}}^{(j_1)}+\mathbf{h}_{j_1}^{(j_{t-1})}$ \\
\hline  $\vdots$ & $\vdots$ &$\ddots$ &$\vdots$ \\
\hline
$\mathbf{h}_{j_0}^{(j_{t-1})}\boxplus_N \mathbf{h}_{j_{t-1}}^{(j_0)}$  &  $\mathbf{h}_{j_1}^{(j_{t-1})}\boxplus_N \mathbf{h}_{j_{t-1}}^{(j_1)}$
    & $\cdots$  &
 $\mathbf{h}_{j_{t-1}}^{(j_{t-1})}$ \\
\hline
\end{tabular}
\end{center}
\end{table}

Secondly, for each $l\in J$, by $\mathbf{h}_{u}^{(l)}$ ($u\in J$) (at the target nodes) and $\mathbf{f}_{i}^{(l)}$ ($i\in [0, k)\backslash I$) (at the $k-t$ remainder nodes of code $\mathcal{C}_4$  connected), $\mathbf{h}_{u}^{(l)}$, $u\in [0,r)\backslash J$ can be reconstructed by \eqref{Eqn perm h} and the fact that $\mathcal{C}_{1}$ is MDS.

Thirdly, for the data in Table \ref{Table Thm MDS 3},  by Fact \ref{Fac Pairing G}-(ii), one can obtain $\mathbf{h}_{u}^{(l)}$ ($l\in [0,r)\backslash J$, $u\in J$) after cancelling $\mathbf{h}_{l}^{(u)}$ ($u\in J$, $l\in [0,r)\backslash J$) (marked with dash underline) by at most $\frac{3}{2}t(r-t)\alpha'$  XORs. Now, $\mathbf{h}_{u}^{(l)}$, i.e., $\mathbf{g}_{\pi_l(u)}^{(l)}$ is available for all $l\in [0,r)$ and $u\in J$.

\begin{table}[htbp]
\begin{center}
\caption{A part of data at the  target nodes that are connected, \protect\\ where $\mathbf{h}_{j_a}^{(j_b)} ~\diamond~ {\mathbf{h}_{j_a}^{(j_b)}}$ denotes $\mathbf{h}_{j_a}^{(j_b)} \boxplus_N {\mathbf{h}_{j_b}^{(j_a)}}$ if $j_a<j_b$ and $\mathbf{h}_{j_a}^{(j_b)} + {\mathbf{h}_{j_b}^{(j_a)}}$ if $j_a>j_b$}\label{Table Thm MDS 3}
\setlength{\tabcolsep}{4.8pt}
\begin{tabular}{|c|c|c|c|}
\hline   TN $j_0$ & TN $j_1$ & $\cdots$  &TN $j_{t-1}$ \\
\hline
\hline $\mathbf{h}_{j_0}^{(j_t)} ~\diamond~ \dashuline{\mathbf{h}_{j_t}^{(j_0)}}$& $\mathbf{h}_{j_1}^{(j_t)}~\diamond~\dashuline{\mathbf{h}_{j_t}^{(j_1)}}$ & $\cdots$ & $\mathbf{h}_{j_{t-1}}^{(j_t)}~\diamond~\dashuline{\mathbf{h}_{j_t}^{(j_{t-1})}}$\\
\hline $\vdots$  & $\vdots$ & $\ddots$  & $\vdots$ \\
\hline
 $\mathbf{h}_{j_0}^{(j_{r-1})}~\diamond~\dashuline{\mathbf{h}_{j_{r-1}}^{(j_0)}}$  & $\mathbf{h}_{j_1}^{(j_{r-1})}~\diamond~\dashuline{\mathbf{h}_{j_{r-1}}^{(j_1)}}$  &
  $\cdots$ & $\mathbf{h}_{j_{t-1}}^{(j_{r-1})}~\diamond~\dashuline{\mathbf{h}_{j_{r-1}}^{(j_{t-1})}}$ \\
\hline
\end{tabular}
\end{center}
\end{table}

Finally, together with $\mathbf{f}_{i}^{(l)}$, $i\in [0, k)\backslash I$,  for each $l\in [0,r)$, the remaining data $\mathbf{f}_{i_0}^{(l)},\cdots,\mathbf{f}_{i_{t-1}}^{(l)}$ can be reconstructed according to the MDS property of  $\mathcal{C}_{1}$.

From the above analysis, we see that reconstructing the original file from $k-t$ remainder nodes and $t$ target nodes requires invoking the reconstruction process of the base code $r$ times and at most $\frac{\alpha'}{2}(3rt-t^2-2t)$ additional XORs, where $0\le t\le \min\{r,k\}$.

\end{proof}

Next, we focus on the repair of the target nodes of  code $\mathcal{C}_4$.

\begin{Theorem}\label{Thm repair parity}
The target nodes of code $\mathcal{C}_4$ have optimal rebuilding access. More specifically, the repair of a target node requires invoking a reconstruction (or encoding) process of the base code and an extra of $\frac{5}{2}(r-1)\alpha'$ XORs.
\end{Theorem}
\begin{proof}
For any given $j\in [0,r)$, in the following, we prove that target node $j$ can be repaired by accessing and
downloading $\mathbf{h'}^{(j)}_{l}$, $l\in [0,r)\backslash\{j\}$, and $\mathbf{f}^{(j)}_{i}$, $i\in [0,k)$.

Firstly, one can compute $\mathbf{g}^{(j)}_{s}$, $s\in [0,r)$ from $\mathbf{f}^{(j)}_{i}$, $i\in [0,k)$ according to the MDS property of the base code, and then obtain $\mathbf{h}^{(j)}_{s}$, $s\in [0,r)$ by \eqref{Eqn perm h}.
Secondly, for any $0\le l\ne j<r$, according to the downloaded data $\mathbf{h'}^{(j)}_{l}=\mathbf{h}^{(j)}_{l}\boxplus_N \mathbf{h}^{(l)}_{j}$ (if $l<j$) or $\mathbf{h'}^{(j)}_{l}=\mathbf{h}^{(j)}_{l}+ \mathbf{h}^{(l)}_{j}$ (if $l>j$) and Fact \ref{Fac Pairing G}-(ii),
$\mathbf{h}^{(l)}_{j}$ is available after cancelling $\mathbf{h}^{(j)}_{l}$ from $\mathbf{h'}^{(j)}_{l}$,  and thus one can get $\mathbf{h'}^{(l)}_{j}=\mathbf{h}^{(l)}_{j}+ \mathbf{h}^{(j)}_{l}$
(if $l<j$) or $\mathbf{h'}^{(l)}_{j}=\mathbf{h}^{(l)}_{j}\boxplus_N \mathbf{h}^{(j)}_{l}$ (if $l>j$). Now target node $j$ is regenerated by noting that  $\mathbf{h'}^{(j)}_{j}=\mathbf{h}^{(j)}_{j}$. The repair of target node $j$ requires invoking a reconstruction (or encoding) process of the base code and an additional $\frac{5}{2}(r-1)\alpha'$ XORs according to Fact \ref{Fac Pairing G}.
\end{proof}
Finally,  we analyze the  repair of the remainder nodes of code $\mathcal{C}_4$.
\begin{Theorem}\label{Thm repair sys}
For each $i\in [0,k)$, remainder node $i$ of the $[n,k]$ MDS array code $\mathcal{C}_4$  maintains the same  normalized repair
bandwidth and rebuilding access as those of  the base code if the repair strategy for remainder node $i$ of the base code is naive, or
during the repair process of node $i$ of the base code, either the following condition (i) or (ii) holds:
\begin{itemize}
    \item [(i)] R1 holds for $j\in [0, r)$, and either R2 or R3 holds;
    \item [(ii)]  $r=2$, R1 holds for $j=0$ or $j=1$, and R2 holds,
\end{itemize}
where R1-R3 are defined as follows.
\begin{itemize}
  \item [R1.] The repair matrix $S_{i,k+j}$ is a block diagonal matrix of the form
  \begin{equation}\label{Eqn_S_matrix}
\hspace{-7mm} blkdiag\begin{pmatrix}S_{i,k+j,0} \hspace{-1mm}& S_{i,k+j,0} \hspace{-1mm}& \cdots \hspace{-1mm}&  \hspace{-1mm}S_{i,k+j,\delta-1} \hspace{-1mm}& S_{i,k+j,\delta-1}\end{pmatrix} 
    \end{equation}
    where $S_{i,k+j,m}$ is an $\frac{N}{2}\times \frac{N}{2}$ matrix (can also be a zero matrix) and is repeated twice in the above matrix for every $m\in [0, \delta)$;
  \item [R2.]
 $\pi_l(j)=\pi_j(l)$ for $j,l\in [0,r)$;
   \item [R3.] $S_{i,k,m}=S_{i,k+j,m}$ for all $j\in[1,r)$ and $m\in [0, \delta)$.
\end{itemize}
In addition, repairing remainder node $i$ of the code $\mathcal{C}_4$ requires invoking the repair strategy of the base code $r$ times, and an extra of $(r-1)\alpha'$ XORs if the repair strategy for remainder node $i$ of the base code is non-naive.
\end{Theorem}

\begin{proof}
If the  repair strategy for  remainder node $i$ of the base code is naive, then the same holds for the new
code $\mathcal{C}_4$. Herein we only now focus on the non-naive case.

Suppose condition (i) holds, then for any $i\in[0,k)$,~$l\in[0,r)$, according to the repair strategy of the base code and R1, we have that  $\mathbf{f}_i^{(l)}$ can be repaired by downloading  $ S_{i,s}\mathbf{f}_s^{(l)}$, $s\in [0,k)\backslash\{i\}$, and 
\begin{equation*}
\left(
                          \begin{array}{cc}
                            S_{i,k+j,m} &  \\
                             & S_{i,k+j,m} \\
                          \end{array}
                        \right)\mathbf{g}_{j}^{(l)}[m], j\in [0,r)     
\end{equation*} from the  surviving nodes, where $m\in[0,\delta)$.

If R2 or R3 holds,  then we have
\begin{equation}\label{Eqn constant S}
S_{i,k+\pi_l(j),m}=S_{i,k+\pi_j(l),m}
\end{equation}
for all $l,j\in [0,r)$ with $l\ne j$ and $m\in [0,\delta)$.
Then, the repair process for remainder node $i$ of code $\mathcal{C}_4$ can be proceeded as follows:
\begin{itemize}
\item [(a)] Download
  $S_{i,s}
\mathbf{f}_{s}^{(l)}$ and 
\begin{equation*}
\left(
                          \begin{array}{cc}
                            S_{i,k+\pi_l(j),m} &  \\
                             & S_{i,k+\pi_l(j),m} \\
                          \end{array}
                        \right)\mathbf{h'}_{j}^{(l)}[m]    
\end{equation*} for all $s\in [0,k)\backslash\{i\}$,  $j,l\in [0,r)$, and $m\in[0,\delta)$,
\item [(b)]For all $l,j\in [0,r)$ with $l\ne j$, by  Fact \ref{Fac Pairing G}, \eqref{Eqn constant S}, and  from
\begin{equation*}
\left(
                          \begin{array}{cc}
                            S_{i,k+\pi_l(j),m} &  \\
                             & S_{i,k+\pi_l(j),m} \\
                          \end{array}
                        \right)\mathbf{h'}_{j}^{(l)}[m]
\end{equation*}
and
\begin{equation*}
\left(
                          \begin{array}{cc}
                            S_{i,k+\pi_j(l),m} &  \\
                             & S_{i,k+\pi_j(l),m} \\
                          \end{array}
                        \right)\mathbf{h'}_{l}^{(j)}[m],
\end{equation*}
we can get
      \begin{equation*}
\left(
                          \begin{array}{cc}
                            S_{i,k+\pi_l(j),m} &  \\
                             & S_{i,k+\pi_l(j),m} \\
                          \end{array}
                        \right)\mathbf{h}_{j}^{(l)}[m] 
\end{equation*}
and
      \begin{equation*}
 \left(
                          \begin{array}{cc}
                            S_{i,k+\pi_j(l),m} &  \\
                             & S_{i,k+\pi_j(l),m} \\
                          \end{array}
                        \right)\mathbf{h}_{l}^{(j)}[m]
 \end{equation*}
  by $2N$ XORs according to Fact \ref{Fac Pairing G}-(ii).
By \eqref{Eqn perm h}, and together with the data
\begin{align*}
&\left(
                          \begin{array}{cc}
                            S_{i,k+\pi_l(l),m} &  \\
                             & S_{i,k+\pi_l(l),m} \\
                          \end{array}
                        \right)\mathbf{h'}_{l}^{(l)}[m]\\=&\left(
                          \begin{array}{cc}
                            S_{i,k+\pi_l(l),m} &  \\
                             & S_{i,k+\pi_l(l),m} \\
                          \end{array}
                        \right)\mathbf{h}_{l}^{(l)}[m]
\end{align*}
obtained in the previous step, the data \begin{equation*}
\left(
                          \begin{array}{cc}
                            S_{i,k+j,m} &  \\
                             & S_{i,k+j,m} \\
                          \end{array}
                        \right)\mathbf{g}_{j}^{(l)}[m], m\in [0,\delta),   
\end{equation*} i.e., $S_{i,k+j}\mathbf{g}_{j}^{(l)}$, is available now for all $l,j\in [0,r)$ by a total of $r(r-1)N$ XORs.
  \item [(c)] From the  data obtained in the previous two steps, the repair procedure of the base MDS code can be invoked to regenerate $\mathbf{f}_i^{(l)}$ for $l\in [0,r)$.
\end{itemize}
Therefore, repairing remainder node $i$ requires invoking the repair strategy of the base code $r$ times and an extra of $(r-1)\alpha'$ XORs.

If condition (ii) holds, the proof proceeds in the same fashion, thus we omit it.
\end{proof}

Note that if condition (ii) of Theorem \ref{Thm repair sys} does not hold, then condition (i) has to be fulfilled to apply Theorem \ref{Thm repair sys}. In this case, for simplicity of verification, one could only check R1 when applying Theorem \ref{Thm repair sys}, as Requirement R2 of Theorem \ref{Thm repair sys} can  be easily  satisfied,
e.g.,  set
$\pi_l(j)=(l+j)\bmod r$, for $l,j\in [0,r)$. Nevertheless, this does not mean that R3 of Theorem \ref{Thm repair sys} is unnecessary. If R1 and R3 hold, then there is no restriction on the permutations in Step 2 as condition (i) of Theorem \ref{Thm repair sys} is now satisfied naturally. Therefore, the permutations can be arbitrary including the identity permutation, i.e., Step 2 of the generic transformation can be omitted, which can simplify the transformation. 

\begin{figure*}
\hrulefill
\begin{equation}\label{Eqn_g'}
\mathbf{v}_{t}^{(l)}=\left(\hspace{-1mm}
                        \begin{array}{c}
                          \mathbf{v}_{t}^{(l)}[0]\\
                          \mathbf{v}_{t}^{(l)}[1] \\
                          \vdots \\
                          \mathbf{v}_{t}^{(l)}[\delta-1] \\
                        \end{array}
                      \hspace{-1mm}\right),~~
\mathbf{v}_{t}^{(l)}[i]=\left\{\hspace{-1mm} \begin{array}{ll}
\mathbf{h}_{t}^{(l)}[i], & \textrm{if}~t=l, \\
\left(
  \begin{array}{l}
    \mathbf{h}_{t}^{(l)}[i,0]+\mathbf{h}_{t}^{(l)}[i,1]+\mathbf{h}_{l}^{(t)}[i,0]+\mathbf{h}_{l}^{(t)}[i,1] \\
    \mathbf{h}_{t}^{(l)}[i,0]+\mathbf{h}_{l}^{(t)}[i,0]\\
  \end{array}
\right)
,& \textrm{if}~t>l,\\
\left(
  \begin{array}{c}
    \mathbf{h}_{t}^{(l)}[i,0]+\mathbf{h}_{t}^{(l)}[i,1]+\mathbf{h}_{l}^{(t)}[i,1] \\
    \mathbf{h}_{l}^{(t)}[i,0]+\mathbf{h}_{l}^{(t)}[i,1]+\mathbf{h}_{t}^{(l)}[i,0] \\
  \end{array}
\right)
, &  \textrm{if}~t<l.
\end{array}\right.
\end{equation}
\end{figure*}

\begin{Remark}
Note that for a binary MDS array code with sub-packetization level $\alpha'$, there may be several pairs of $(\delta, N)$ such that $\alpha'=\delta N$ and $N$ is even. To apply the generic transformation, one should choose an appropriate pair of $(\delta, N)$ so that either condition (i) or condition (ii) of Theorem \ref{Thm repair sys} can be satisfied.
\end{Remark}

To this end, we give a concluding result of this subsection according to Theorems \ref{Thm MDS C4}-\ref{Thm repair sys}.

\begin{Theorem}\label{Thm_Sum}
For an $[n, k]$ binary MDS array code with sub-packetization level $\alpha'=\delta N$, by applying the generic transformation in Section \ref{subsec generic trans}, a new $[n, k]$ binary MDS array code with sub-packetization level $\alpha=r\alpha'$ can be obtained such that
\begin{itemize}
    \item Any arbitrarily chosen $n-k$ nodes have  optimal rebuilding access. 
    \item The remaining nodes can maintain the
  same  normalized repair
bandwidth and rebuilding access as those of  the base code if the repair strategy for the remainder nodes of the base code is naive, or if
either condition (i) or (ii) in Theorem \ref{Thm repair sys} holds.
\end{itemize}
\end{Theorem}

\subsection{An Alternative Pairing Technique for Step 3 -  Target Nodes Unchanged}\label{subsection:step2'}

Note that the code obtained by the generic transformation in Section \ref{subsec generic trans} is no longer of the systematic form if some $r$ systematic nodes of a systematic base code are chosen to be the target nodes.
Although it is a well-known fact that a linear code in non-systematic form can always be converted into a systematic one, such a conversion is not trivial/easy in general since normally one needs to solve a linear equation system. To this end, here we provide a substitute pairing operation that can maintain the systematic form of the code without performing further conversion, which can also avoid the complexity caused by the conversion.

W.L.O.G., the last $r$ nodes are chosen as the target nodes while the data modification is carried out at the first $r$ nodes. Note that $\mathbf{f}_{0}^{(l)}, \cdots, \mathbf{f}_{r-1}^{(l)}$ can be represented by $\mathbf{f}_{r}^{(l)}, \cdots, \mathbf{f}_{k-1}^{(l)}, \mathbf{g}_{0}^{(l)}, \cdots$, $\mathbf{g}_{r-1}^{(l)}$ for any $l\in [0,r)$ by the MDS property of the base code $\mathcal{C}_{1}$. That is,
$$\mathbf{f}_{j}^{(l)} \hspace{-.5mm}=\hspace{-.5mm} \sum\limits_{t=r}^{k-1}A_{j,t}\mathbf{f}_{t}^{(l)}+\hspace{-.5mm}\sum\limits_{t=0}^{r-1} A_{j,t}\mathbf{g}_{t}^{(l)}\hspace{-.5mm}=\hspace{-.5mm}\sum\limits_{t=r}^{k-1}A_{j,t}\mathbf{f}_{t}^{(l)}+\hspace{-.5mm}\sum\limits_{t=0}^{r-1} A_{j,\pi_l(t)}\mathbf{h}_{t}^{(l)}$$
where $A_{j,0},\cdots,A_{j,k-1}$ are some nonsingular matrices of order $N$. Then, define a new  storage code $\mathcal{C}_4'$  (see Table \ref{Table C3'})
as
\begin{equation}\label{Eqn f'}
  \mathbf{f'}_{j}^{(l)} = \sum\limits_{t=r}^{k-1}A_{j,t}\mathbf{f}_{t}^{(l)}+\sum\limits_{t=0}^{r-1} A_{j,\pi_l(t)}\mathbf{v}_{t}^{(l)},  ~j,~l\in [0,r),
\end{equation}
where $\mathbf{v}_{t}^{(l)}$ is defined in \eqref{Eqn_g'}.
\begin{table}[htbp]
\begin{center}
\caption{Structure of the new storage code  $\mathcal{C}_4'$}\label{Table C3'}
\setlength{\tabcolsep}{2.3pt}
\begin{tabular}{|c|c|c|c|c|c|c|c|c|}
\hline  RN 0  & $\cdots$  & RN   $r-1$   &RN   $r$ & $\cdots$  & RN   $k-1$ & TN   $0$ &$\cdots$  & TN   $r-1$ \\
\hline\hline   $\mathbf{f'}_0^{(0)}$ & $\cdots$ &$\mathbf{f'}_{r-1}^{(0)}$ &$\mathbf{f}_{r}^{(0)}$ & $\cdots$ & $\mathbf{f}_{k-1}^{(0)}$ & $\mathbf{h}_{0}^{(0)}$ & $\cdots$ & $\mathbf{h}_{r-1}^{(0)}$ \\
\hline   $\mathbf{f'}_0^{(1)}$ & $\cdots$  & $\mathbf{f'}_{r-1}^{(1)}$ &$\mathbf{f}_{r}^{(1)}$&$\cdots$  &$\mathbf{f}_{k-1}^{(1)}$&  $\mathbf{h}_{0}^{(1)}$& $\cdots$  &$\mathbf{h}_{r-1}^{(1)}$ \\
\hline   $\vdots$ & $\ddots$ & $\vdots$ & $\vdots$ &$\ddots$ &$\vdots$ &$\vdots$ & $\ddots$ & $\vdots$ \\
\hline
$\mathbf{f'}_0^{(r-1)}$ &   $\cdots$&   $\mathbf{f'}_{r-1}^{(r-1)}$ & $\mathbf{f}_{r}^{(r-1)}$  &  $\cdots$& $\mathbf{f}_{k-1}^{(r-1)}$ &  $\mathbf{h}_{0}^{(r-1)}$&  $\cdots$& $\mathbf{h}_{r-1}^{(r-1)}$ \\
\hline
\end{tabular}
\end{center}
\end{table}




\begin{figure*}
\hrulefill
\begin{equation}\label{Eqn_S_i,j,w}
S_{i,j,w}^{(t+1)}= \left\{
           \begin{array}{ll}
             \mathbf{I}_{\alpha'/2},  & \mathrm{if}~w_t\equiv i\,(\bmod\,r)  \mathrm{~and~}tr\le k-r, \mathrm{or~}w_{\lfloor\frac{k}{r}\rfloor}= i-(k-r) \mathrm{~and~}tr> k-r\\
             \mathbf{0}_{\alpha'/2},  & \mathrm{otherwise}
          \end{array}
         \right.
\end{equation}
\end{figure*}

Note from \eqref{Eqn_g'}  that for $t,l\in [0,r)$ with $t>l$, we have
\begin{equation*}
 \mathbf{v}_{t}^{(l)}[i]+ \mathbf{v}_{l}^{(t)}[i]=\mathbf{h}_{t}^{(l)}[i]\mbox{~~and~~} \mathbf{v}_{l}^{(t)}[i]\boxplus \mathbf{v}_{t}^{(l)}[i]=\mathbf{h}_{l}^{(t)}[i],
\end{equation*}
which together with \eqref{Eqn Oplus} we further have
 $\mathbf{v}_{t}^{(l)}+ \mathbf{v}_{l}^{(t)}=\mathbf{h}_{t}^{(l)}$ and  $\mathbf{v}_{l}^{(t)}\boxplus_N \mathbf{v}_{t}^{(l)}=\mathbf{h}_{l}^{(t)}$.
That is,  the new code $\mathcal{C}_4'$ can be obtained by applying Step 3 to the code $\mathcal{C}_3'$ (cf. Table \ref{Table C2 II}), where $$\mathbf{f'}_0^{(l)}, \cdots, \mathbf{f'}_{r-1}^{(l)}, \mathbf{f}_{r}^{(l)}, \cdots, \mathbf{f}_{k-1}^{(l)}, \mathbf{v}_{0}^{(l)}, \cdots, \mathbf{v}_{r-1}^{(l)}$$ is an instance of  $\mathcal{C}_1$ by \eqref{Eqn f'}.

\begin{table}[htbp]
\begin{center}
\caption{The storage  code  $\mathcal{C}_3'$}\label{Table C2 II}
\setlength{\tabcolsep}{2.4pt}
\begin{tabular}{|c|c|c|c|c|c|c|c|c|}
\hline RN 0  & $\cdots$  & RN   $r-1$   &RN   $r$ & $\cdots$  & RN   $k-1$ & TN   $0$ &$\cdots$  & TN   $r-1$ \\
\hline\hline $\mathbf{f'}_0^{(0)}$ & $\cdots$ &$\mathbf{f'}_{r-1}^{(0)}$ &$\mathbf{f}_{r}^{(0)}$ & $\cdots$ & $\mathbf{f}_{k-1}^{(0)}$ & $\mathbf{v}_{0}^{(0)}$ & $\cdots$ & $\mathbf{v}_{r-1}^{(0)}$ \\
\hline   $\mathbf{f'}_0^{(1)}$ & $\cdots$  & $\mathbf{f'}_{r-1}^{(1)}$ &$\mathbf{f}_{r}^{(1)}$&$\cdots$  &$\mathbf{f}_{k-1}^{(1)}$&  $\mathbf{v}_{0}^{(1)}$& $\cdots$  &$\mathbf{v}_{r-1}^{(1)}$ \\
\hline   $\vdots$ & $\ddots$ & $\vdots$ & $\vdots$ &$\ddots$ &$\vdots$ &$\vdots$ & $\ddots$ & $\vdots$ \\
\hline
  $\mathbf{f'}_0^{(r-1)}$ &
   $\cdots$&
   $\mathbf{f'}_{r-1}^{(r-1)}$ &
 $\mathbf{f}_{r}^{(r-1)}$  &
  $\cdots$&
 $\mathbf{f}_{k-1}^{(r-1)}$ &
  $\mathbf{v}_{0}^{(r-1)}$&
  $\cdots$&
 $\mathbf{v}_{r-1}^{(r-1)}$ \\
\hline
\end{tabular}
\end{center}
\end{table}

Similarly to \cite{transform-IT}, we  immediately have the following  result.

\begin{Theorem}
 The code $\mathcal{C}_4'$ has the MDS property and the same repair property as that of the code $\mathcal{C}_4$.
\end{Theorem}

\section{Applications and comparisons}


In this section, based on the transformation in Section \ref{Sec tran}, 
we propose two specific  applications.  The first application is  recursively applying the transformation multiple times to an arbitrary binary MDS array code to obtain a new one with optimal rebuilding access for all nodes. The second application is applying the transformation one time to an arbitrary binary MDS array codes with optimal repair bandwidth/optimal rebuilding access for all systematic nodes to obtain a new one with the corresponding optimality property for all nodes. Finally, we provide extensive comparisons between the transformation proposed in this paper and the one in \cite{HouLee}.

\subsection{Application From an Arbitrary Binary MDS Array Code}

Choosing any $[n,k]$ binary    MDS array code with even sub-packetization level $\alpha'$ as the base code $\mathcal{Q}_{0}$, for example, one can choose the EVENODD code  \cite{EVENODD} or its generalizations in \cite{Blaum_Bruck_Vardy96,EVENODDG}, or the STAR code \cite{STAR} as the base code. We can build a binary MDS array code with optimal rebuilding access for all nodes through   Algorithm \ref{alg1}.
\begin{algorithm}[htbp]
  \caption{} \label{alg1}
  \begin{algorithmic}[1]
\State Given an $[n,k]$ binary  MDS array code $\mathcal{Q}_{0}$ with even sub-packetization level $\alpha'$, let $m=\lceil n/r\rceil$ where $r=n-k$; 
  	\For{$i=0$; $i<m$; $i++$}
       \State Set the  code $\mathcal{Q}_{i}$ as the base code;
       \If{$i<m-1$}
       \State Designate nodes $\min\{ir, k-r\}, \min\{ir, k-r\}+1, \cdots, \min\{ir, k-r\}+r-1$ as the target nodes;
       \Else
       \State Designate nodes $k, k+1, \cdots, n-1$ as the target nodes;
       \EndIf
       \State Let $N=\alpha'$ and $\delta=r^i$, applying the generic transformation to the base code $\mathcal{Q}_{i}$ to get a new MDS array code $\mathcal{Q}_{i+1}$ with sub-packetization level $r^{i+1}\alpha'$;
        \EndFor
  \end{algorithmic}  	
\end{algorithm}

According to Theorems \ref{Thm MDS C4}-\ref{Thm repair sys}, Algorithm \ref{alg1} eventually gives a binary MDS array code $\mathcal{Q}_{m}$ with optimal rebuilding access, where the sub-packetization level is $r^{\lceil n/r\rceil}\alpha'$. Specifically, applying Algorithm \ref{alg1} to the $[n=k+2, k]$ EVENODD code with sub-packetization level $\alpha'=k-1$ in \cite{EVENODD}, where $k$ is an odd prime, we can obtain an $[n=k+2, k]$ MDS array code with optimal rebuilding access for all nodes, where the sub-packetization level is $(k-1)2^{\frac{n+1}{2}}$, which is $k-1$ times the lower bound of the minimal required sub-packetization level derived in \cite{tight-bound-on-alpha}. However, the lower bound in \cite{tight-bound-on-alpha} is not sensitive to the underlying finite field.

Strictly speaking, to recursively apply Theorem \ref{Thm repair sys}, one needs to verify that requirement R1 of Theorem \ref{Thm repair sys}  is satisfied during each round of the transformation in Algorithm \ref{alg1}. Indeed, it can be checked by induction as follows.

\noindent\textbf{Induction hypothesis}:
For $t\in[1, m-1)$, the first $tr$ nodes  of the code $\mathcal{Q}_{t}$ can be optimally repaired and the $r^t\alpha'\times r^t\alpha'$ repair matrix $S_{i,j}^{(t)}$ of  node $i$ ($i\in [0, tr)$)  is of the form
\begin{equation}\label{Eqn_Alg1_repair_1}
 S_{i,j}^{(t)}=blkdiag\begin{pmatrix}\hspace{-1mm}
  S_{i,j,0}^{(t)} \hspace{-1mm}&
   S_{i,j,0}^{(t)}\hspace{-1mm}&
 \hspace{-1mm}  \cdots\hspace{-1mm}&
   S_{i,j,r^{t}-1}^{(t)}\hspace{-1mm}&
   S_{i,j,r^{t}-1}^{(t)}\hspace{-1mm}
 \end{pmatrix}
\end{equation}
and
\begin{eqnarray}\label{Eqn_Alg1_repair_2}
S_{i,j,w}^{(t)}= \left\{
           \begin{array}{ll}
             \mathbf{I}_{\alpha'/2},  & \mathrm{if}~w_l\equiv i\,(\bmod\,r)\\
             \mathbf{0}_{\alpha'/2},  & \mathrm{otherwise}
          \end{array}
         \right.
\end{eqnarray}
for all  $j\in [0,n)\backslash\{i\}$ and $w\in [0,r^{t})$ with $(w_{t-1},\cdots,w_1,w_0)$ denoting its $r$-ary expansion, where $l$ is a nonnegative  integer such that $lr\le i<(l+1)r$.

In the $t$-th round, the MDS array code $\mathcal{Q}_t$ is set as the base code, and nodes
\begin{equation*}
\min\{tr, k-r\}, \min\{tr, k-r\}+1, \cdots, \min\{tr, k-r\}+r-1
\end{equation*}
are designated as target nodes. Clearly,  the repair matrices of the first $\min\{tr, k-r\}$ nodes of the base code $\mathcal{Q}_t$ satisfy the requirement R1 of Theorem \ref{Thm repair sys} by  \eqref{Eqn_Alg1_repair_1} and \eqref{Eqn_Alg1_repair_2}.  Further,  the induction hypothesis is true for all the $r^{t+1}\alpha'\times r^{t+1}\alpha'$ repair matrices  $S_{i,j}^{(t+1)}$ of the first $\min\{(t+1)r, k\}$  nodes  of the MDS array code $\mathcal{Q}_{t+1}$ since
\begin{equation*}
S_{i,j}^{(t+1)}= blkdiag\begin{pmatrix}
S_{i,j}^{(t)}& \cdots & S_{i,j}^{(t)}
\end{pmatrix}
\end{equation*}
for  $i\in [0,\min\{tr, k-r\}$), $j\in [0,n)\backslash\{i\}$,
and
\begin{align*}
&S_{i,j}^{(t+1)}\\=& blkdiag\begin{pmatrix}
\hspace{-1mm}S_{i,j,0}^{(t+1)} & \hspace{-1mm}S_{i,j,0}^{(t+1)} & \hspace{-1mm}\cdots & \hspace{-1mm}S_{i,j,r^{t+1}-1}^{(t+1)}&\hspace{-1mm}S_{i,j,r^{t+1}-1}^{(t+1)}
\end{pmatrix}
\end{align*}
for $i\in [\min\{tr, k-r\},\min\{(t+1)r, k\})$ and  $j\in [0,n)\backslash\{i\}$
with $S_{i,j,w}^{(t+1)}$ being defined in \eqref{Eqn_S_i,j,w} according to Theorem \ref{Thm repair parity}.

In what follows, we give an example of the  application.

\begin{Example}
For a  $[5,3]$  EVENODD code $\mathcal{Q}_0$, which is binary and has a sub-packetization level of $2$,  the structure of $\mathcal{Q}_0$ can be
depicted as in Table \ref{Ex Table C1} \cite{EVENODD}.

\begin{table}[htbp]
\begin{center}
\caption{A $[5,3]$ EVENODD code $\mathcal{Q}_0$, where SN and PN denote systematic node and parity node, respectively}\label{Ex Table C1}
\begin{tabular}{|c|c|c||c|c|}
\hline
SN $0$  &  SN $1$  &  SN $2$   &PN $0$ &PN $1$  \\
\hline\hline
$a_0$  & $b_{0}$& $c_{0}$&  $a_0+b_0+c_0$& $a_0+b_1+c_0+c_1$ \\
\hline
$a_1$  & $b_{1}$& $c_{1}$&  $a_1+b_1+c_1$& $a_1+b_0+b_1+c_0$ \\
\hline
\end{tabular}
\end{center}
\end{table}

In the following, through three rounds of transformations according to Algorithm \ref{alg1}, we convert $\mathcal{Q}_0$ into an MDS array code with optimal rebuilding access for all nodes. We obtain codes $\mathcal{Q}_1$, $\mathcal{Q}_2$ and $\mathcal{Q}_3$, where $\mathcal{Q}_1$ and $\mathcal{Q}_2$ are shown  in Tables \ref{Ex Table Q2} and \ref{Ex Table Q3}, respectively, while the representation of $\mathcal{Q}_3$ is omitted due to the limited space. For simplicity,
all the permutations in Step 2 of each round
are chosen to be the identity permutation, as R3 of  Theorem \ref{Thm repair sys} is obviously satisfied for the non-naively repaired remainder nodes of the intermediate code $\mathcal{Q}_i$ ($i\in [1,m)$).

\begin{table}[htbp]
\begin{center}
\caption{A $[5,3]$ binary MDS array code $\mathcal{Q}_1$, where systematic nodes 0 and 1 are chosen as the target nodes}\label{Ex Table Q2}
\setlength{\tabcolsep}{1.1pt}
\begin{tabular}{|c|c|c||c|c|}
\hline
SN $0$  &  SN $1$  &  SN $2$   &PN $0$ &PN $1$  \\
\hline\hline
$a_0$  & $b_{0}$& $c_{0}$&  $a_0\hspace{-.3mm}+\hspace{-.3mm}a_2\hspace{-.3mm}+\hspace{-.3mm}a_3\hspace{-.3mm}+\hspace{-.3mm}b_0\hspace{-.3mm}+\hspace{-.3mm}b_1\hspace{-.3mm}+\hspace{-.3mm}c_0$& $a_0+a_2+b_0+c_0+c_1$ \\
\hline
$a_1$  & $b_{1}$& $c_{1}$&  $a_1+a_2+b_0+c_1$& $a_1+a_3+b_1+c_0$ \\
\hline\hline
$a_2$  & $b_{2}$& $c_{2}$&  $a_2+a_3+b_1+b_2+c_2$& $a_2\hspace{-.3mm}+\hspace{-.3mm}a_3\hspace{-.3mm}+\hspace{-.3mm}b_1\hspace{-.3mm}+\hspace{-.3mm}b_3\hspace{-.3mm}+\hspace{-.3mm}c_2\hspace{-.3mm}+\hspace{-.3mm}c_3$ \\
\hline
$a_3$  & $b_{3}$& $c_{3}$&  $a_2+b_0+b_1+b_3+c_3$& $a_2\hspace{-.3mm}+\hspace{-.3mm}b_0\hspace{-.3mm}+\hspace{-.3mm}b_1\hspace{-.3mm}+\hspace{-.3mm}b_2\hspace{-.3mm}+\hspace{-.3mm}b_3\hspace{-.3mm}+\hspace{-.3mm}c_2$ \\
\hline
\end{tabular}
\end{center}
\end{table}

\begin{table*}[htbp]
\begin{center}
\caption{The $[5,3]$ binary MDS array code $\mathcal{Q}_2$ in systematic form, where systematic nodes 1 and 2 are chosen as the target nodes}\label{Ex Table Q3}
\begin{tabular}{|c|c|c||c|c|}
\hline
SN $0$  &  SN $1$  &  SN $2$   &PN $0$ &PN $1$  \\
\hline\hline
$a_0$  & $b_{0}$& $c_{0}$&  $a_0+a_2+a_3+b_0+b_1+b_4+b_5+c_0+c_1$& $a_0+a_2+b_0+b_5+c_1$ \\
\hline
$a_1$  & $b_{1}$& $c_{1}$&  $a_1+a_2+b_0+b_4+c_0$& $a_1+a_3+b_1+b_4+b_5+c_0+c_1$ \\
\hline
$a_2$  & $b_{2}$& $c_{2}$&  $a_2+a_3+b_1+b_2+b_6+b_7+c_2+c_3$& $a_2+a_3+b_1+b_3+b_7+c_3$ \\
\hline
$a_3$  & $b_{3}$& $c_{3}$&  $a_2+b_0+b_1+b_3+b_6+c_2$& $a_2+b_0+b_1+b_2+b_3+b_6+b_7+c_2+c_3$ \\
\hline\hline
$a_4$  & $b_{4}$& $c_{4}$&  $a_4+a_6+a_7+b_5+c_0+c_4$& $a_4+a_6+b_4+b_5+c_1+c_4+c_5$ \\
\hline
$a_5$  & $b_{5}$& $c_{5}$&  $a_5+a_6+b_4+b_5+c_1+c_5$& $a_5+a_7+b_4+c_0+c_1+c_4$ \\
\hline
$a_6$  & $b_{6}$& $c_{6}$&  $a_6+a_7+b_4+b_6+b_7+c_0+c_1+c_3+c_6$& $a_6+a_7+b_4+b_6+c_0+c_1+c_2+c_3+c_6+c_7$ \\
\hline
$a_7$  & $b_{7}$& $c_{7}$&  $a_6+b_5+b_6+c_0+c_2+c_3+c_7$& $a_6+b_5+b_7+c_0+c_2+c_6$ \\
\hline
\end{tabular}
\end{center}
\end{table*}

Note that during the first two rounds of transformations, the alternative pairing technique  in Section \ref{subsection:step2'}  is employed in Step 3, i.e., data modification is carried out at the parity nodes. For example, when applying the generic transformation to the code $\mathcal{Q}_0$,  after Steps 1 and 2, we get an intermediate code  $\mathcal{Q}'_{0}$ as in Table \ref{Ex Table Q1'}, while in Step 3, $b_0,~b_1,~a_2$ and $a_3$ in the parity nodes are respectively replaced by
  $a_2+a_3+b_0+b_1,~    a_2+b_0,~a_2+a_3+b_1,~a_2+b_0+b_1$
according to \eqref{Eqn_g'},
which  leads to the code $\mathcal{Q}_{1}$ in Table \ref{Ex Table Q2}.

\begin{table}[htbp]
\begin{center}
\caption{An intermediate code  $\mathcal{Q}'_{0}$, where SN and PN respectively denote systematic node and parity node}\label{Ex Table Q1'}
\begin{tabular}{|c|c|c||c|c|}
\hline
SN $0$  &  SN $1$  &  SN $2$   &PN $0$ &PN $1$  \\
\hline\hline
$a_0$  & $b_{0}$& $c_{0}$&  $a_0+b_0+c_0$& $a_0+b_1+c_0+c_1$ \\
\hline
$a_1$  & $b_{1}$& $c_{1}$&  $a_1+b_1+c_1$& $a_1+b_0+b_1+c_0$ \\
\hline
\hline
$a_2$  & $b_{2}$& $c_{2}$&  $a_2+b_2+c_2$& $a_2+b_3+c_2+c_3$ \\
\hline
$a_3$  & $b_{3}$& $c_{3}$&  $a_3+b_3+c_3$& $a_3+b_2+b_3+c_2$ \\
\hline
\end{tabular}
\end{center}
\end{table}

It is seen that $\mathcal{Q}_2$ maintains the MDS property. Furthermore, the first four nodes can be respectively optimally repaired by accessing and downloading symbols in rows $\{1,2,5,6\}$, $\{3,4,7,8\}$, $\{1,2,3,4\}$,  $\{5,6,7,8\}$  of Table
\ref{Ex Table Q3}
from all the surviving nodes.

\end{Example}

\subsection{Application From  Binary MDS array Codes With   Optimal Repair   for Systematic Nodes}

In the literature, there are some  binary MDS array codes with optimal repair bandwidth for all the systematic nodes but not all the parity nodes. For example, the $[n=k+2,k]$ MDR code in \cite{MDR} and the two modified
versions \cite{MDR-M}.
In this subsection, we propose the second application of the generic transformation by  presenting Algorithm \ref{alg2}, which can  build a binary MDS array code with optimal repair bandwidth  for all nodes from  one with optimal repair bandwidth for all the systematic nodes but not all the parity nodes.

\begin{algorithm}[htbp]
  \caption{} \label{alg2}
  \begin{algorithmic}[1]
\State Choosing any $[n=k+r,k]$ binary MDS array code $\mathcal{C}_1$  with optimal repair bandwidth for all the systematic nodes but not   all the parity nodes as the base code, let $\alpha'$ denote its sub-packetization level;
\State Designate the $k$ systematic nodes of the   MDS array code $\mathcal{C}_1$ as the remainder nodes while the $r$ parity nodes as the target nodes;
\If {either condition (i) or condition (ii) of Theorem \ref{Thm repair sys} holds for the MDS array code $\mathcal{C}_1$ for some $\delta$ and $N$ with $\alpha'=\delta N$}
\State Let  $\mathcal{C}_2=\mathcal{C}_1$; 
\Else
\State  Combine two instances of the MDS array code $\mathcal{C}_1$ to get a new   code $\mathcal{C}_2$ with sub-packetization level $\alpha''=2\alpha'$; 
\State Set $\delta=1$ and $N=\alpha''$;
\EndIf
\State Applying the generic transformation in Section \ref{subsec generic trans} to the MDS array code $\mathcal{C}_2$ to generate a new MDS array code $\mathcal{C}_3$;
  \end{algorithmic}  	
\end{algorithm}
\begin{table*}[htbp]
\begin{center}
\caption{The  $[6,4]$  MDR-1 code $\mathcal{C}_1$ over $\mathbf{F}_2$}\label{Ex Table MDR0}
\begin{tabular}{|c|c|c|c||c|c|}
\hline
SN $0$  &  SN $1$  &  SN $2$   & SN $3$   &PN $0$ &PN $1$  \\
\hline\hline
$a_0$  & $b_{0}$& $c_{0}$& $d_{0}$&  $a_0+b_0+c_0+d_0$& $a_0+a_3+b_0+c_1+c_4+d_4$ \\
\hline
$a_1$  & $b_{1}$& $c_{1}$& $d_{1}$& $a_1+b_1+c_1+d_1$& $a_1+a_2+b_1+c_5+d_0+d_5$ \\
\hline
$a_2$  & $b_{2}$& $c_{2}$& $d_{2}$& $a_2+b_2+c_2+d_2$& $a_2+b_1+b_2+c_6+d_3+d_6$ \\
\hline
$a_3$  & $b_{3}$& $c_{3}$& $d_{3}$& $a_3+b_3+c_3+d_3$& $a_3+b_0+b_3+c_2+c_7+d_7$ \\
\hline
$a_4$  & $b_{4}$& $c_{4}$& $d_{4}$& $a_4+b_4+c_4+d_4$& $a_0+a_4+a_7+b_0+b_4+c_0+c_5+d_0$ \\
\hline
$a_5$  & $b_{5}$& $c_{5}$&  $d_{5}$&$a_5+b_5+c_5+d_5$& $a_1+a_5+a_6+b_1+b_5+c_1+d_1+d_4$ \\
\hline
$a_6$  & $b_{6}$& $c_{6}$& $d_{6}$& $a_6+b_6+c_6+d_6$& $a_2+a_6+b_2+b_5+b_6+c_2+d_2+d_7$ \\
\hline
$a_7$  & $b_{7}$& $c_{7}$& $d_{7}$& $a_7+b_7+c_7+d_7$& $a_3+a_7+b_3+b_4+b_7+c_3+c_6+d_3$ \\
\hline
\end{tabular}
\end{center}
\end{table*}

Note that in Algorithm \ref{alg2}, if  neither condition (i) nor condition (ii)  of Theorem \ref{Thm repair sys} holds for  $\mathcal{C}_1$, then for $i\in [0,~k)$ and $j\in [0,~r)$,  the repair matrix   $S_{i,k+j}^{\mathcal{C}_2}$  of systematic node $i$ of $\mathcal{C}_2$ has the form
\begin{eqnarray*}
S_{i,k+j}^{\mathcal{C}_2}=\left(
                            \begin{array}{cc}
                              S_{i,k+j}^{\mathcal{C}_1} &  \\
                               & S_{i,k+j}^{\mathcal{C}_1} \\
                            \end{array}
                          \right),
\end{eqnarray*}
according to Line 6, thus R1 of Theorem \ref{Thm repair sys} holds  for $j\in [0,~r)$, and therefore condition (i) holds, where $S_{i,k+j}^{\mathcal{C}_1}$  denotes  the repair matrix of systematic node $i$ of $\mathcal{C}_1$. Thus, Algorithm \ref{alg2} generates  an $[n=k+r,k]$ binary MDS array code with optimal repair bandwidth for all nodes by Theorems \ref{Thm MDS C4}-\ref{Thm repair sys}, while the sub-packetization level is  $2\alpha'$ or $4\alpha'$.

In what follows, we give an example by performing  the application to  the
first version of the modified MDR (MDR-1 for short) code in \cite{MDR-M}, where the two parity nodes of the code are designated as the target nodes.  For the
$[n=k+2,k]$ MDR-1 code with sub-packetization level $\alpha'=2^{k-1}$ in \cite{MDR-M},  we focus on the general case $k\ge 4$, and 
let us first recall the repair strategy of the systematic nodes. 
For an integer $j$, where $j\in [0,2^{k-1})$, let $(j_{k-2},\cdots,j_1, j_{0})$ be its binary expansion, i.e., $j=j_0+2j_1+\cdots+2^{k-2}j_{k-2}$.
When repairing node $i$ ($i\in [0,k)$) of the  MDR-1 code in \cite{MDR-M}, one accesses and downloads the $j$-th element from each surviving node for all  $j\in J_i$, where
{\small
\begin{equation*}
J_i=\left\{ \begin{array}{ll}
\{j=(j_{k-1},\cdots,j_1,j_0)|j_1=i\}, & \textrm{if}~i=0, 1, \\
\{j=(j_{k-1},\cdots,j_1,j_0)|j_0+j_1\in\{0,2\}\}, &  \textrm{if}~i=2,\\
\{j=(j_{k-1},\cdots,j_1,j_0)|j_0+j_1\in\{1,3\}\}, &  \textrm{if}~i=3,\\
\{j=(j_{k-1},\cdots,j_1,j_0)|j_{i-2}=1\}, &  \textrm{if}~3<i<k.\\
\end{array}\right.
\end{equation*}
}
\begin{Example}
For the  $[6,4]$ MDR-1 code $\mathcal{C}_1$ in \cite{MDR-M}, which is binary and has a sub-packetization level of $8$,  the structure of $\mathcal{C}_1$ can be depicted as in Table \ref{Ex Table MDR0}. By applying   the second application to  code $\mathcal{C}_1$,
we can get a $[6,4]$ binary MDS array code $\mathcal{C}_3$ with optimal rebuilding access for all nodes, as shown in Table \ref{Ex Table MDR3}, where $\textbf{f}_0^{(l)}, \textbf{f}_1^{(l)}, \textbf{f}_2^{(l)}, \textbf{f}_3^{(l)}, \textbf{g}_0^{(l)}$, $\textbf{g}_1^{(l)}$ denote an instance of  the  $[6,4]$  MDR-1 code $\mathcal{C}_1$ for $l\in [0,4)$ with $\textbf{g}_0^{(l)}, \textbf{g}_1^{(l)}$ denoting the parity data.


\begin{table}[htbp]
\begin{center}
\caption{A $[6,4]$ binary MDS array code $\mathcal{C}_3$}\label{Ex Table MDR3}
\begin{tabular}{|c|c|c|c||l|c|}
\hline
SN $0$  &  SN $1$  & SN $2$  &  SN $3$  & PN $0$ &PN $1$  \\
\hline\hline
 $\textbf{f}_0^{(0)}$& $\textbf{f}_1^{(0)}$&  $\textbf{f}_2^{(0)}$& $\textbf{f}_3^{(0)}$ &  $\textbf{g}_0^{(0)}$& $\textbf{g}_1^{(0)}+\textbf{g}_0^{(2)}$ \\
\hline
  $\textbf{f}_0^{(1)}$& $\textbf{f}_1^{(1)}$&  $\textbf{f}_2^{(1)}$& $\textbf{f}_3^{(1)}$ &  $\textbf{g}_0^{(1)}$& $\textbf{g}_1^{(1)}+\textbf{g}_0^{(3)}$ \\
\hline
 $\textbf{f}_0^{(2)}$& $\textbf{f}_1^{(2)}$&  $\textbf{f}_2^{(2)}$& $\textbf{f}_3^{(2)}$ &  $\textbf{g}_0^{(2)}+\textbf{g}_1^{(0)}+\textbf{g}_1^{(1)}$& $\textbf{g}_1^{(2)}$ \\
\hline
 $\textbf{f}_0^{(3)}$& $\textbf{f}_1^{(3)}$& $ \textbf{f}_2^{(3)}$& $\textbf{f}_3^{(3)}$ &  $\textbf{g}_0^{(3)}+\textbf{g}_1^{(0)}$& $\textbf{g}_1^{(3)}$ \\
\hline
\end{tabular}
\end{center}
\end{table}

For the code $\mathcal{C}_3$, one can directly verify that the code is MDS. Furthermore,
systematic node $i$ ($i\in [0,4)$) is optimally repaired by accessing and downloading the $j$-th element from the surviving nodes for all  $j\in J_i$, where
\begin{equation*}
J_i\hspace{-.3mm}=\hspace{-.4mm}\left\{\hspace{-2.5mm} \begin{array}{ll}
\{j=(j_4,j_3,j_{2},j_1,j_0)|j_1=i\},& \textrm{if}~i=0, 1, \\
\{j=(j_4,j_3,j_{2},j_1,j_0)|j_0+j_1\in\{0,2\}\},&  \textrm{if}~i=2,\\
\{j=(j_4,j_3,j_{2},j_1,j_0)|j_0+j_1\in\{1,3\}\},&  \textrm{if}~i=3.
\end{array}\right.
\end{equation*}

Parity nodes $0, 1$ can be optimally repaired by accessing and downloading the $j$-th element from the surviving nodes for all  $j\in J_4$ and $J_5$, respectively, where
\begin{equation*}
 J_4=\{j=(j_4,j_3,j_{2},j_1,j_0)|j_4=0\}
\end{equation*} 
and
\begin{equation*}
J_5=\{j=(j_4,j_3,j_{2},j_1,j_0)|j_4=1\}.
\end{equation*}
\end{Example}

\begin{Remark}\label{Rem1}
In Algorithm \ref{alg2}, we can also choose any binary MDS array code  with an efficient but not optimal repair for   systematic nodes as the base code, then Algorithm \ref{alg2} generates  another binary MDS array code with optimal rebuilding access for the parity nodes, while the repair efficiency for the systematic nodes is maintained as in the base code. For example, we can apply Algorithm \ref{alg2} to the $[n=p+2, k=p]$ EVENODD code\footnote{The efficient repair of the systematic nodes follows from the strategy in \cite{wang2010rebuilding}.} with sub-packetization level $p-1$ in \cite{EVENODD}, where $p$ is a prime. 
Whereas, the transformation in  \cite{HouLee}  is only applicable to the  $[p+2, p]$ EVENODD code with $4\mid (p-1)$.
\end{Remark}

\subsection{Comparisons Between the Proposed Transformation and the One in \cite{HouLee}}\label{sec:comp}

The transformation in \cite{HouLee} also relies on a pairing technique and the data stored in each node is regarded as a polynomial, the two   different linear combinations of polynomials $a(x)$ and $b(x)$ are generated as $a(x)+b(x)$ and $a(x)+\delta(x)b(x)$ for some polynomial $\delta(x)$.
Compared to the one in \cite{HouLee},   the generic transformation proposed in this paper has the following advantages:
\begin{itemize}
    \item \textbf{The transformation in this paper
is uniform}. It can be applied to any binary MDS array codes such as the MDS codes in \cite{EVENODD,Blaum_Bruck_Vardy96,EVENODDG,STAR,MDR,MDR-M,EEGad} through a unified pairing technique. Whereas the pairing technique  in the transformation in \cite{HouLee} is not uniform, i.e., the coefficient polynomial $\delta(x)$ is not uniform but should be chosen case by case such that the statements (which were only proved in general in \cite{HouLee} if the base code is the EVENODD code) still hold for the chosen base code.   Whether an appropriate coefficient polynomial $\delta(x)$ can always be found for any existing binary MDS array codes is unknown in \cite{HouLee}.

\item \textbf{The computation complexity of the proposed transformation is very low.} To apply the transformation, only multiple linear operations between two binary column vectors of length $2$ 
is introduced to the base code,   whereas, for the   transformation in \cite{HouLee}, the extra computation cost  introduced to the base code is a linear  operation between  two polynomials in $\mathbf{F}_{2}[x]$ with degree equals to $\alpha'-1$ (where $\alpha'$ denotes the sub-packetization level of the base code), which  needs more XORs and thus results in a higher computation complexity.

For the MDS array codes obtained by the transformation proposed in this work and the one in \cite{HouLee}, 
the reconstruction process and repair process of a target node of the new code are mainly relying on those processes of the base code and some additional XORs.  Therefore, it suffices to compare the extra complexity involved, i.e., the numbers of additional XORs that are needed, which mainly depend on the pairing technique involved. The pairing technique is independent of the base code for the transformation proposed in this work, but depends on the base code for the one in \cite{HouLee}. However, in \cite{HouLee}, the pairing technique is only given to the EVENODD codes in general in Section II, therefore, we take this pairing technique for a quantified comparison of the extra computational complexity introduced to the transformation proposed in this work and the one in \cite{HouLee} in
Table \ref{Table comp_quan}. We see that the transformation proposed in this work indeed has a much lower computational complexity than the one in \cite{HouLee}.

\begin{table*}[htbp]
\begin{center}
\caption{Comparison of the (extra) complexity between the transformations obtained in this paper and the one in \cite{HouLee} (based on the 1st transformation in Section II), where  $t$ denotes the number of target nodes that are connected during the reconstruction
}\label{Table comp_quan}
\begin{tabular}{|c|c|c|c|}
\hline     & Complexity for Pairing &  Extra complexity for reconstruction & Extra complexity for repairing TN  \\
\hline\hline The transformation in \cite{HouLee}   &  $2r(r-1)\alpha'$ XORs &
$\frac{1}{2}(6rt-3t^2-3t)\alpha'$ XORs  &    $4(r-1)\alpha'$ XORs\\
\hline  The proposed transformation   &   $\frac{5r(r-1)}{4}\alpha'$ XORs& $\frac{1}{2}(3rt-t^2-2t)\alpha'$ XORs &  $\frac{5}{2}(r-1)\alpha'$ XORs \\
\hline
\end{tabular}
\end{center}
\end{table*}
\begin{table*}[htbp]
\begin{center}
\caption{Comparison between the transformations obtained in this paper and the one in \cite{HouLee}}\label{Table comp}
\begin{tabular}{|c|c|c|c|}
\hline     & Pairing  technique involved&  Computation complexity &Codes can be applied   \\
\hline\hline The transformation in \cite{HouLee}   &  Case-by-case &
High  &   Some specific codes\\
\hline  The proposed transformation   &   Uniform& Low  &  Any binary MDS array codes\\
\hline
\end{tabular}
\end{center}
\end{table*}


\item \textbf{The proposed transformation has a wider range of applications in the case of $d=n-1$.} It can be applied to any binary MDS array codes, whereas, the applications of the transformation in \cite{HouLee} are limited, whether it can be applied to binary MDS array codes with   optimal repair bandwidth such as the codes in \cite{MDR,MDR-M} is unknown. Particularly, when applying the transformation in \cite{HouLee} one time to the $[n=p+r, k=p]$ EVENODD code to enable optimal rebuilding access for all the parity while keeping the repair efficiency of the systematic nodes, it is valid only if $4\mid (p-1)$ and $r=2$, while ours does not have such restrictions (see Remark \ref{Rem1}).
\end{itemize}
Table \ref{Table comp} summarizes the above comparison.


\begin{Remark}
The transformation in \cite{HouLee}   also works for binary MDS array codes with $d<n-1$, however, when repairing a failed node of the resultant code,  some of the $d$ helper nodes should be specifically selected. Note that by operating on $d-k+1$ instances of the base code in the first step, the newly proposed transformation can be easily generalized to the case of $d<n-1$ with a restricted selection of helper nodes. Thus, instead of involving the imperfect and complicated case of $d<n-1$ in this paper, we would rather  leave the case of $d<n-1$ and   the $d$ helper nodes can be arbitrarily chosen  in  our future work to keep the simplicity of the current paper.
\end{Remark}

\section{Conclusion}
In this work, we proposed a generic transformation that can be applied to any binary MDS array code to produce new binary MDS array codes with some arbitrarily chosen $r$ nodes having optimal repair bandwidth and optimal rebuilding access. Besides, based on this transformation, we provided  two useful applications to yield binary MDS array codes with optimal repair bandwidth/rebuilding access for all nodes. The computation complexity of the  new binary MDS array codes obtained by our generic transformation  is very low,  with only simple XOR operations needed to  perform on the base code to complete the encoding,  decoding, and repair processes.  Extending the transformation to the case of $d<n-1$ with the $d$ helper nodes that can be arbitrarily chosen  is part of our ongoing work.

\section*{Acknowledgment}
The authors would like to thank the Associate Editor Prof. Ido Tal and the three anonymous reviewers for their valuable suggestions and comments, which have greatly improved the presentation and quality of this paper.


\ifCLASSOPTIONcaptionsoff
  \newpage
\fi

%
\begin{IEEEbiography}[{\includegraphics[width=1in,height=1.25in,clip,keepaspectratio]{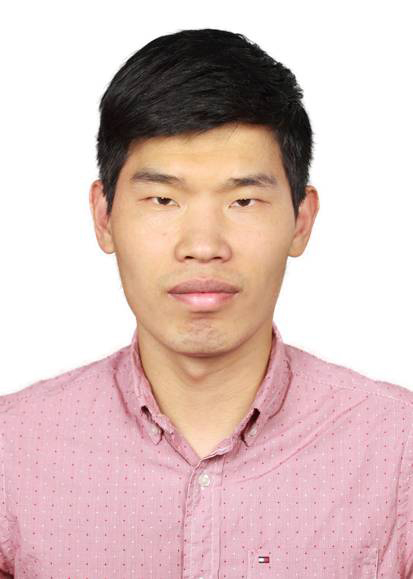}}]
{Jie Li}(Member, IEEE)    
received the B.S. and M.S. degrees in mathematics from Hubei University, Wuhan, China, in 2009 and 2012, respectively, and received the Ph.D. degree from the department of communication engineering, Southwest Jiaotong University, Chengdu, China, in 2017. From  2015 to 2016, he was a visiting Ph.D. student with the Department of Electrical Engineering and Computer Science, The University of Tennessee at Knoxville, TN, USA.  From   2017 to   2019, he was a postdoctoral researcher with the Department of Mathematics, Hubei University, Wuhan, China. From  2019 to 2021, he was a postdoctoral researcher with the Department of Mathematics and Systems Analysis, Aalto University, Finland. He is currently a senior researcher with the Theory Lab, Huawei Tech. Investment Co., Limited, Hong Kong SAR, China. His research interests include private information retrieval, coding for distributed storage, and sequence design.

Dr. Li received the IEEE Jack Keil Wolf ISIT Student Paper Award in 2017.
\end{IEEEbiography}

\begin{IEEEbiography}[{\includegraphics[width=1in,height=1.25in,clip,keepaspectratio]{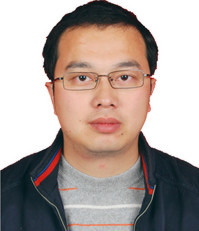}}]
{Xiaohu Tang}(Senior Member, IEEE)  
 received the B.S. degree in applied
mathematics from Northwestern Polytechnical University, Xi’an, China,
in 1992, the M.S. degree in applied mathematics from Sichuan University,
Chengdu, China, in 1995, and the Ph.D. degree in electronic engineering
from Southwest Jiaotong University, Chengdu, in 2001.

From 2003 to 2004, he was a Research Associate with the Department of
Electrical and Electronic Engineering, The Hong Kong University of Science and Technology. From 2007 to 2008, he was a Visiting Professor with the University of Ulm, Germany. Since 2001, he has been with the School of Information Science and Technology, Southwest Jiaotong University, where he is currently a Professor. His research interests include coding theory, network security, distributed storage, and information processing for big data.

Dr. Tang was a recipient of the National Excellent Doctoral Dissertation Award, China, in 2003, the Humboldt Research Fellowship, Germany, in 2007, and the Outstanding Young Scientist Award by NSFC, China, in 2013. He served as an Associate Editors for several journals, including IEEE TRANSACTIONS ON INFORMATION THEORY and \textit{IEICE Transactions on
Fundamentals}, and served for a number of technical program committees
of conferences.
\end{IEEEbiography}

\begin{IEEEbiography}[{\includegraphics[width=1in,height=1.25in,clip,keepaspectratio]{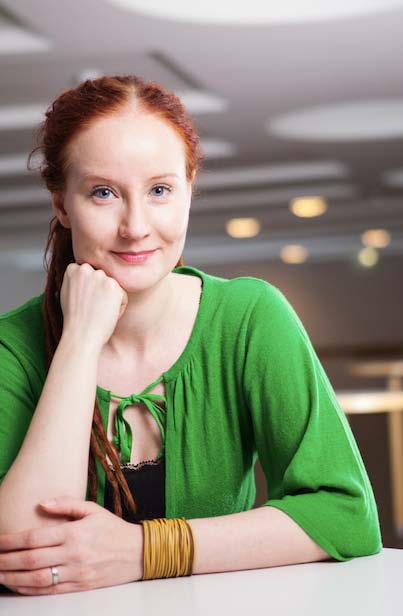}}]
{Camilla Hollanti}(Member, IEEE)     received the M.Sc. and Ph.D. degrees from the University of Turku, Finland, in 2003 and 2009, respectively, both in pure mathematics. Her research interests lie within applications of algebraic number theory to wireless communications and physical layer security, as well as in combinatorial and coding theoretic methods related to distributed storage systems and private information retrieval.

For 2004-2011 Hollanti was with the University of Turku. She joined the University of Tampere as  Lecturer for the academic year 2009-2010. Since 2011, she has been with the Department of Mathematics and Systems Analysis at Aalto University, Finland, where she currently works as Full Professor and Vice Head, and leads a research group in Algebra, Number Theory, and Applications. During 2017-2020, Hollanti was affiliated with the Institute of Advanced Studies at the Technical University of Munich, where she held a three-year Hans Fischer Fellowship, funded by the German Excellence Initiative and the EU 7th Framework Programme.

Hollanti is currently an editor of the AIMS Journal on Advances in Mathematics of Communications,  SIAM Journal on Applied Algebra and Geometry, and IEEE Transactions on Information Theory. She is a recipient of several grants, including six Academy of Finland grants. In 2014, she received the World Cultural Council Special Recognition Award for young researchers. In 2017, the Finnish Academy of Science and Letters awarded her the V\"ais\"al\"a Prize in Mathematics. For 2020-2022, Hollanti is serving as a member of the Board of Governors of the IEEE Information Theory Society, and is one of the General Chairs of IEEE ISIT 2022.
\end{IEEEbiography}




\end{document}